\RequirePackage{rotating}
\documentclass[oneside,letterpaper,10pt]{article}
\pdfoutput=1

\usepackage[letterpaper,total={6.5in, 9in},top=1in]{geometry}

\def\SidewayTables{1}
\usepackage{multirow}
\usepackage{xspace}
\usepackage{pifont}
\usepackage[plain]{fancyref}
\usepackage{flushend}
\usepackage[nounderscore]{syntax}
\usepackage{graphicx} % Extended graphics package.
\usepackage{amsmath} % American Mathematics Society standards
\usepackage{amsxtra} % Additional math symbols
\usepackage{amssymb} % Additional math symbols
\usepackage{amsthm} % Additional math symbols
\usepackage{latexsym} % Additional math symbols
\usepackage{setspace} % Controls line spacing
\usepackage[titles]{tocloft} % Adds leader dots to all entries in the table of contents
\usepackage[final]{listings}
\usepackage{mathpartir}
\newtheorem{theorem}{Theorem}[section]
\newtheorem{lemma}[theorem]{Lemma}
\usepackage{longtable}
\usepackage[draft=true]{minted}
\usepackage{hyperref}
\usepackage{booktabs}
\hypersetup{
    linktocpage=true,
    colorlinks=true,
    citecolor=blue,
    urlcolor=blue,
    linkcolor=blue,
    citebordercolor={1 0 0},
    urlbordercolor={1 0 0},
    linkbordercolor={.7 .8 .8},
    breaklinks=true,
    }
    
\newcommand{\eg}{e.g.,\xspace}
\newcommand{\ie}{i.e.,\xspace}
\newcommand{\etal}{et al.\xspace}

\newenvironment{customTable}{
  \ifdef{\SidewayTables}{\begin{sidewaystable}}{\begin{table*}}
}{
  \ifdef{\SidewayTables}{\end{sidewaystable}}{\end{table*}}
}

\newcommand{\app}[3]{\textit{appId} \hspace{1mm} #1 \hspace{1mm} #2 \hspace{1mm} #3}

\newcommand{\srule}[3]{\langle #1, #2 \rangle \rightarrow #3}

\newcommand{\st}[2]{(#1,#2)}

\newcommand{\screen}[4]{\textbf{screen} \hspace{1mm} #1 \hspace{1mm} #2 \hspace{1mm} #3 \hspace{1mm} #4}

\newcommand{\ps}[3]{\textbf{proxy} \hspace{1mm} #1 \hspace{1mm} #2 \hspace{1mm} #3}

\newcommand{\widget}[3]{#1 \hspace{1mm} #2 \hspace{1mm} #3}

\newcommand{\trans}[5]{\textbf{transition} \hspace{1mm} #1 \hspace{1mm} \textbf{dest} \hspace{1mm} #2 \hspace{1mm} (\textbf{cond} \hspace{1mm} #3 \hspace{1mm} \textbf{and} \hspace{1mm} #4) \hspace{1mm} #5}

\newcommand{\param}[2]{\textbf{param} \hspace{1mm} #1 \hspace{1mm} #2}

\newcommand{\bop}[3]{#1 \hspace{1mm} \textbf{#2} \hspace{1mm} #3}

\newcommand{\fn}[3]{\textbf{fun} \hspace{1mm} \textit{#1} \hspace{1mm} #2 \hspace{1mm} #3}

\newcommand{\dash}{\hspace{1mm}|\hspace{1mm}}

\newcommand{\ruleExp}[3]{\langle #1, \hspace{1mm}#2 \rangle \rightarrow #3}

\newcommand{\ruleExpL}[3]{\langle (#1), \hspace{1mm}#2 \rangle \rightarrow^* #3}

\newcommand{\ruleExpK}[4]{\langle (#1), \hspace{1mm}#2 \rangle \rightarrow^{#4} #3}

\newcommand{\keyword}[1]{\textbf{#1}}

\begin{document}

\title{SeMA: Extending and Analyzing Storyboards to Develop Secure Android Apps}

\author{
Joydeep Mitra\\
Kansas State University, USA\\
joydeep@ksu.edu 
\and
Venkatesh-Prasad Ranganath\\
rvprasad.free@gmail.com
\and
Torben Amtoft\\
Kansas State University, USA\\
tamtoft@ksu.edu
\and
Michael Higgins\\
Kansas State University, USA\\
mikehiggins@ksu.edu
}

\maketitle

\begin{abstract}
  Mobile apps provide various critical services, such as banking, communication, and healthcare. To this end, they have access to our personal information and have the ability to perform actions on our behalf. Hence, securing mobile apps is crucial to ensuring the privacy and safety of its users. 
  
  Recent research efforts have focused on developing solutions to secure mobile ecosystems (\ie app platforms, apps, and app stores), specifically in the context of detecting vulnerabilities in Android apps. Despite this attention, known vulnerabilities are often found in mobile apps, which can be exploited by malicious apps to harm the user. Further, fixing vulnerabilities after developing an app has downsides in terms of time, resources, user inconvenience, and information loss. 
  
  In an attempt to address this concern, we have developed SeMA, a mobile app development methodology that builds on existing mobile app design artifacts such as storyboards. With SeMA, security is a first-class citizen in an app’s design -- app designers and developers can collaborate to specify and reason about the security properties of an app at an abstract level without being distracted by implementation level details.  Our realization of SeMA using Android Studio tooling demonstrates the methodology is complementary to existing design and development practices.  An evaluation of the effectiveness of SeMA shows the methodology can detect and help prevent 49 vulnerabilities known to occur in Android apps.  Further, a usability study of the methodology involving ten real-world developers shows the methodology is likely to reduce the development time and help developers uncover and prevent known vulnerabilities while designing apps.
\end{abstract}

\section{Introduction}
% Android apps have become an integral aspect of modern-day living. With the growing use of these apps, it is vital to secure them. Existing approaches to secure Android apps are curative, \ie they are aimed at detecting vulnerabilities after they are introduced. Identifying and fixing vulnerabilities after they occur increase the cost of development. Consequently, there has been a growing call to integrate security into every phase of the software development life-cycle. In a 2014 study, Green and Smith argued that to build secure software, developers need support in various areas ranging from safer programming languages to better security testing tools.~\cite{Green:2016} Therefore, given the current landscape of Android app security and to support the call for secure software development, there is scope for exploring an alternative approach aimed at \textit{preventing} vulnerabilities from occurring as opposed to \textit{curing} them.

\subsection{Motivation}
Android apps help us with critical tasks such as banking and communication. Consequently, they need access to our personal information. Hence, developers try to ensure that their apps do not cause harm to the user. Despite such efforts, apps exhibit vulnerabilities that can be exploited by malicious apps either installed in the device or executing remotely. In 2018, an industrial study of 17 fully functional Android apps discovered that all the apps had vulnerabilities ~\cite{ptSec:URL}. 43\% of the vulnerabilities were classified as high risk, and 60\% of them were on the client-side. In 2019, a vulnerability in Google's camera app allowed a malicious app without required permissions to gain full control of the camera app and access the photos and videos stored by the camera app ~\cite{checkmarx:URL}. A recent study showed that approximately 11\% of 2,000 Android apps collected from a couple of app stores, including Google Play, are vulnerable to Man-In-The-Middle (MITM) and phishing attacks ~\cite{Yingjie:ACM19}. Further, a vulnerability in numerous real-world apps was recently discovered (\eg the GMail app on Android) that allowed a malicious app to exploit the vulnerability and carry out phishing and denial-of-service attacks on those apps ~\cite{TaskAffinityCVE:URL}.

In the last decade, researchers have developed a plethora of tools and techniques to help detect vulnerabilities in Android apps ~\cite{Sufatrio:2015, sadeghi:2017, Li:2017}. Even so, apps with vulnerabilities find their way to app stores because app developers do not use these tools, or the tools are ineffective. The latter reason is supported by Ranganath and Mitra ~\cite{Ranganath:EMSE19}.  Additionally, Pauck \etal ~\cite{Pauck:2018} assessed six prominent taint analysis tools aimed at discovering vulnerabilities in Android and found tools to be ineffective in detecting vulnerabilities in real-world apps. In a similar vein, Luo \etal ~\cite{Luo:2019} conducted a qualitative analysis of Android static taint analysis tools and found that these tools not only miss certain taint flows since they do not consider an app's context, but their default configurations lead to many false positives.

Moreover, existing approaches to secure Android apps are curative; that is, they are aimed at detecting vulnerabilities after they are introduced during implementation. Identifying and fixing vulnerabilities after they occur in implementation increases the cost of development. Consequently, there has been a growing call to integrate security into every phase of the software development life-cycle. In a 2014 study, Green and Smith ~\cite{Green:2016} argued that to build secure software, developers need support in various areas ranging from safer programming languages to better security testing tools. Therefore, given the current landscape of Android app security and to support the call for secure software development, we are exploring a preventive approach that can help prevent the occurrence of vulnerabilities in app implementations (as opposed to curing vulnerabilities in app implementations).

\subsection{Contributions}
Motivated by the above observations, we make the following contributions in this article:

\begin{enumerate}
    \item We propose SeMA, a development methodology based on an existing mobile app design technique called \emph{storyboarding}. We describe the methodology along with the techniques used to realize it.
    
    \item We demonstrate the effectiveness of SeMA -- \textit{can SeMA detect and help prevent vulnerabilities?} -- by applying it to 49 known Android app vulnerabilities that are captured in the Ghera benchmark suite and are likely to occur in real-world apps.  Our findings show SeMA is highly effective in detecting these vulnerabilities.

    \item We demonstrate the usability of SeMA by having developers build Android apps with features commonly found in real-world apps.  Our findings suggest, with basic interventions, SeMA is highly effective in helping developers detect and prevent the vulnerabilities targeted by SeMA without increasing the design and development times.
\end{enumerate}

\emph{In the rest of the paper, all our remarks about SeMA's ability to prevent vulnerabilities are limited to the 49 vulnerabilities targeted by SeMA.}
    
\section{Background}

SeMA borrows heavily from Model Driven development to iteratively build an app from its storyboard, an existing design artifact.

\paragraph{Model Driven Development (MDD)} In MDD, software is developed by iteratively refining abstract and formal models ~~\cite{Brambilla:2017}. A model is meant to capture the application's behavior and is expressed in a domain-specific language (DSL). Apart from models, the domain-specific platform is a crucial entity in MDD. The domain-specific platform provides frameworks and APIs to enable easy refinement of a model into a platform-specific implementation. Since every aspect of an application can seldom be specified in the DSL, the resulting implementation is often extended with additional code; mostly, the business logic of the application. 

Traditionally, the software engineering community has been hesitant to adopt MDD due to social, cultural, and technical reasons ~\cite{Bran:2012}. However, with the relatively recent development of sophisticated tools that enable MDD, software engineering teams are more open to embracing it. For example, Amazon uses TLA+ to develop web services ~\cite{Lamport:2015, Wayne:2018}. Tools like Alloy ~\cite{Jackson:2002}, UML/OCL ~\cite{Richters:2002}, and UMLSec ~\cite{jan:2002} help create and analyze models of software behavior, which form the basis of further development.

\paragraph{Storyboarding} Android app development teams use storyboards to design an app's navigation ~\cite{DesTech:URL,UXMag:URL}. A storyboard is a sequence of screens and transitions between the screens. It serves as a model for the user's observed behavior in terms of screens and transitions. Numerous tools such as Xcode ~\cite{Xcode:URL}, Sketch ~\cite{Sketch:URL}, and Android's JetPack ~\cite{JetPackNav:URL} help express a storyboard digitally. The storyboarding process is \textit{participatory} in nature because it allows designers to get feedback from potential users about the features captured in the storyboard and from developers about the feasibility of implementing those features. However, traditional storyboards cannot capture an app's behavior beyond UI and navigational possibilities. This limitation hinders the possibility of collaboration between designers and developers while specifying the storyboard.   

\paragraph{MDD with Storyboarding} Existing MDD approaches to mobile app development do not use storyboards and do not consider security requirements. Extending storyboards with capabilities to capture an app's behavior to enable MDD has numerous benefits. First, since storyboards are becoming an integral part of the mobile app development process ~\cite{Trapp:2013}, MDD based on storyboarding is not likely to disrupt the current mobile app development process. Second, an extended storyboard can serve as a common substrate for collaboration between designers and developers to specify an app's behavior along with its UI and navigational features.  Third, an extended storyboard can serve as a model of the app to formally analyze its behavior. Further, the model's abstractness reduces the cost of such analysis as it eliminates the need for and the associated cost of extracting high-level (abstract) behaviors captured in code. Finally, the storyboard can serve as a reference for an app's behavior when auditing the app's implementation.

\section{The Methodology}
\label{sec:method}
SeMA enables the reasoning and verification of security properties of an app's storyboard via iterative refinement. A schematic of the process of developing an app with SeMA is shown in \Fref{fig:schematic}. The process begins with a developer sketching the initial storyboard of an app, as shown in Figure \ref{fig:init}. The developer then extends the storyboard with the app's behavior and checks if the behaviors satisfy various pre-defined security properties. The developer may repeat the previous steps to revise and refine the behaviors. Once the storyboard has captured the behavior as intended by the developer while satisfying pre-defined security properties, the developer generates an implementation from the storyboard with a push of a button. As the final step, the developer adds business logic to the implementation via hooks provided in the generated code. 

\begin{figure}
    \centering
    \includegraphics[height=6cm,width=0.5\textwidth]{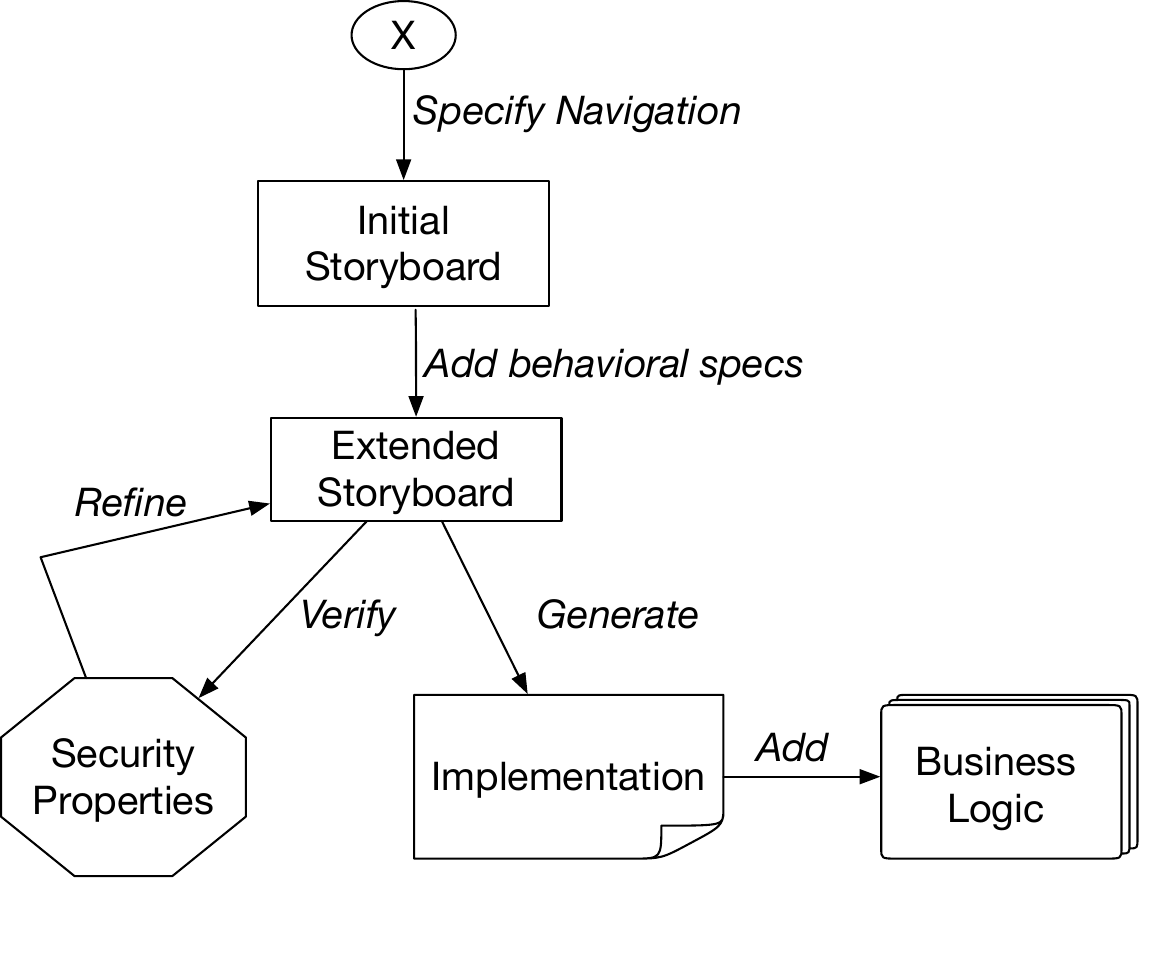}
    \caption{Development Process in SeMA}
    \label{fig:schematic}
\end{figure}

\begin{figure}
    \centering
    \includegraphics[width=\textwidth]{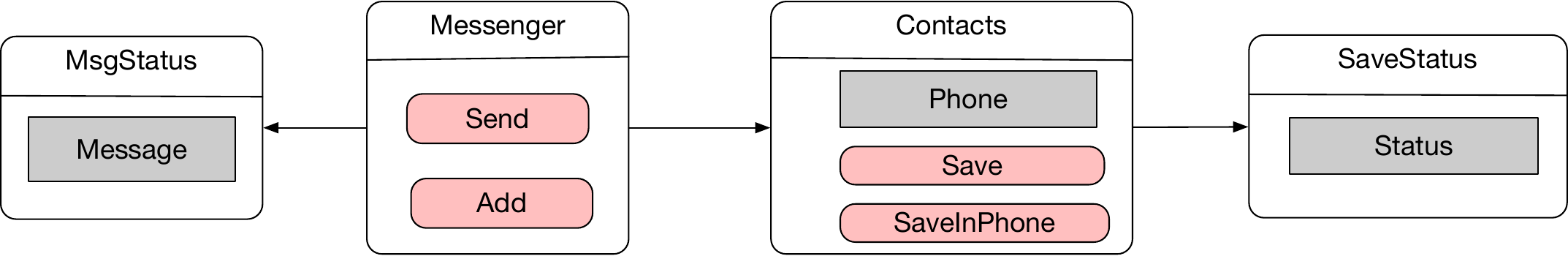}
    \caption{Diagrammatic representation of an initial storyboard of an app in SeMA, which is also a traditional storyboard.}
    \label{fig:init}
\end{figure}

\subsection{Extended Storyboard}
\label{sec:informal-sem}

A \emph{traditional storyboard} used in the design of mobile apps is composed of screens and transitions between screens. A screen is a collection of named widgets that allow the user to interact with the app, \eg clicking a button.  A transition (edge) between two screens depicts a navigation path from the source screen to the destination screen. The basic structure of a storyboard defines the navigational paths in an app. 

For example, \Fref{fig:init} shows the traditional storyboard of an app with four screens: \texttt{Messenger}, \texttt{Contacts}, \texttt{MsgStatus}, and \texttt{SaveStatus}. Starting from the \textit{Messenger} screen, a user can either add contact numbers to the app via \texttt{Contacts} screen or send a message to all saved contact numbers.

A traditional storyboard does not support the specification of app behavior. Hence, we propose the following extensions to traditional storyboards to enable the specification of app behaviors in storyboards (\ie enrich a traditional storyboard, as in \Fref{fig:init}) into an extended storyboard, as in \Fref{fig:story_ex}. \textit{These figures are diagrammatic representations of storyboards (known as  navigation graphs) in Android Studio.}

\begin{figure}
    \centering
    \includegraphics[width=\linewidth]{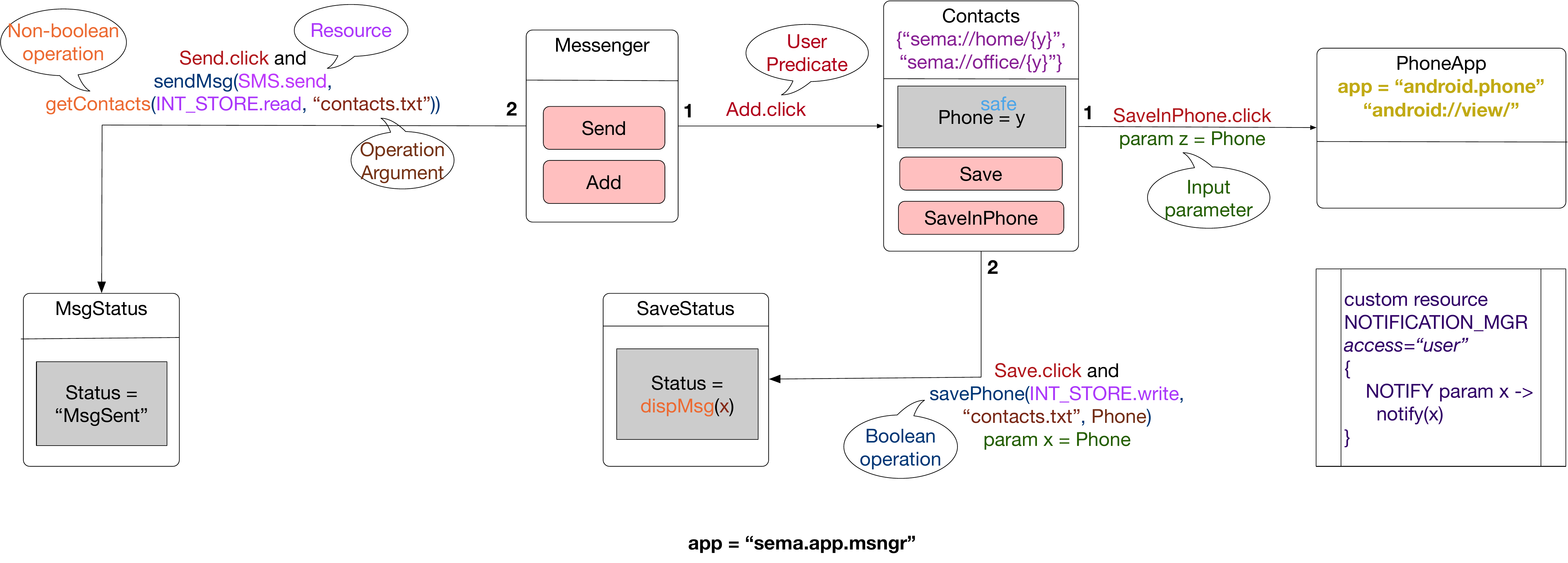}
    \caption{Diagrammatic representation of an extended storyboard. The (information) bubbles are not part of the storyboard.}
    \label{fig:story_ex}
\end{figure}

\paragraph{App Identity} Every storyboard has an \textit{app} attribute, which is a unique string constant used to identify the app described by the storyboard. 

\paragraph{Extensions to Screens} Screens are extended with a mandatory \emph{name}, optional \emph{input parameters} (in green in \Fref{fig:story_ex}), and optional \textit{URIs} (in purple in \Fref{fig:story_ex}). 

The \emph{Input parameters of a screen} are similar to parameters of a function in mainstream programming.  Input parameters bind to the values (arguments) provided when the screen is activated by either another screen in the app via an incoming transition or an app via the screen's \textit{URI}.

The \textit{URIs} associated with a screen can be used to access the screen from outside the app. A URI can have input parameters; similar to query parameters in web URLs. URI input parameters serve as input parameters of the screen. All URIs associated with a screen must have the same set of input parameters. For example, both URIs associated with the \texttt{Contacts} screen in \Fref{fig:story_ex}, have the input parameter \texttt{y}. External apps accessing a screen via its URI must provide the arguments corresponding to the URI's input parameters.  Every URI without its parameters must be unique in an app. 

\paragraph{Proxy Screens of External Apps} Apps often interact with external apps.  To capture this interaction, \emph{proxy screen} representing a screen of an external app can be included in an extended storyboard of an app; see \texttt{PhoneApp} in \Fref{fig:story_ex}. Such proxy screens have a mandatory \textit{name}, a mandatory \textit{URI}, and an optional \emph{app} attributes. If \textit{app} is specified in a proxy screen, then the proxy screen denotes the screen identified by the \textit{URI} in the app named \textit{app}. If \textit{app} is not specified, then the proxy screen denotes any screen identified by the \textit{URI} in an app installed on the device and determined by Android.

% As an aside, input parameter names must be unique across screens. So, if \textit{y} is an input parameter to screen S1, then it cannot be an input parameter to screen S2. This restriction helps the accuracy of the analysis (in Section \ref{sec:analysis}) that tracks the flow of input parameters in a storyboard. \rv{Is this an implementation detail?  If so, defer it to a section that discusses the limitations of current implementation of the methodology}

\paragraph{Extensions to Widgets} Widgets are extended with a mandatory \textit{value} that can be assigned by the developer (\eg in labels), entered by the user (\eg in fields), provided by an input parameter of the containing screen (\eg when the screen is activated by a transition), or returned by an operation. Based on the displayed content, widgets can be of different types, \eg a label and a text widget display plain text while a web widget displays web content. Further, depending on the widget's type, a widget can be configured with a pre-defined set of rules that regulate the data displayed in a widget, \eg a \textit{WebView} widget can be configured with a whitelist of trusted URL patterns (via \textit{trusted-patterns} attribute) to display content only from URLs in the whitelist.

\paragraph{Resources} Android provides apps with resources with different capabilities, \eg storage, networking. Hence, to complement this aspect, storyboards are extended with a pre-defined set of resources with specific capabilities that can be used by the apps being designed. \footnote{The current realization of SeMA supports a subset of pre-defined resources offered by Android.}  Android apps can offer UI-less services to other apps, \eg broadcast receivers, content providers. Such services are denoted by custom resources in storyboards. A \emph{custom resource} offers capabilities that can be used by apps installed on the device. Each custom resource has a mandatory identifier that is unique to the app. Each capability of a resource has a mandatory identifier that is unique to the resource.  Also, each capability that has security implications can be marked as \emph{privileged}.  Access to a resource and its capabilities can be controlled via the \textit{access} attribute of the resource.  This attribute can take on one of the following three values: \textit{all} indicates any app can access the resource/capability, \textit{user} indicates user permission is required to access the resource/capability, and \textit{own} indicates only the resource defining app \textit{x} or an app that shares the digital signature of app \textit{x} can access the resource/capability.  For example, in \Fref{fig:story_ex}, NOTIFICATION\_MGR is a custom resource that offers NOTIFY capability with a \texttt{notify} operation. Based on its access attribute, a client will need to seek the user's permission to use its NOTIFY capability.  These values of the \textit{access} attribute are based on permission levels supported by Android.

\paragraph{Operations} In an extended storyboard, an \textit{operation} indicates a task to be performed, \eg read from a file, get contents from a web server. An operation has a \textit{name}, returns a value, may have \textit{input parameters}, and may use a \textit{capability} (provided by a resource). An operation is used by mentioning its name along with arguments and any required capabilities. For example, in \Fref{fig:story_ex}, operation \texttt{savePhone} is used to save data in the device's internal storage by using the \emph{write} capability of the internal storage device exposed as the resource \textit{INT\_STORE}. A use of an operation introduces (declares) it in the storyboard.  An operation is defined in the generated implementation (described in \Fref{sec:codegen}). Use of operations \textit{must} be consistent; that is, a non-boolean operation cannot be used in a boolean value position.

\paragraph{Extension to Transitions} Transitions between screens can be adorned with constraints that, when satisfied, enable/trigger a transition. A constraint is a conjunction of a \emph{user action} (\eg click of a button) and a \emph{boolean expression}, comprising a combination of boolean operations. A constraint is satisfied when the user action is performed (\eg \emph{save} button is clicked) and the boolean expression evaluates to true.  For example, in \Fref{fig:story_ex}, the transition from \texttt{Contacts} screen to \texttt{SaveStatus} screen is taken only when the user clicks the \texttt{Save} button and the \texttt{savePhone} operation evaluates to \textit{true}; that is, the value of \texttt{Phone} is successfully saved in "contacts.txt", a file on the internal storage of the device.

As part of a transition, arguments are provided to the input parameters of the destination screen (in green in \Fref{fig:story_ex}). An argument can be a literal specified in the storyboard, a value available in the source screen of the transition (\eg value of a contained widget, input parameter to the screen), or a value returned by an operation. Further, every transition to a screen must provide values (arguments) for every input parameter of that screen. For example, if there are two transitions \emph{t1} and \emph{t2} to screen \emph{s} with input parameters \emph{x} and \emph{y}, then arguments for both \emph{x} and \emph{y} must be provided along both \emph{t1} and \emph{t2}.

Multiple outgoing transitions from a screen may be simultaneously enabled when their constraints are not mutually exclusive. Hence, to handle such situations, all outgoing transitions from a screen must be totally ordered; see the number labels on outgoing transitions from \texttt{Contacts} screen in \Fref{fig:story_ex}. The implementation derived from the storyboard will evaluate the constraints according to the specified order of transitions and take the first enabled transition.

The formal description of the syntax and semantics of the extended storyboard is provided in Appendix \ref{sec:syntax} and \ref{sec:semantics}, respectively.

\subsection{Security Properties}
\label{sec:secProps}
An Android app often interacts with other apps on the device, the underlying platform, and remote servers. Such interaction involves sharing of information, responding to events, and performing tasks based on user actions. Many of these interactions have security implications that should be considered during app development, \eg can the user's personal information be stored safely on external storage? (information leak and data injection), who should have access to the content provided by the app? (permissions and privilege escalation), how should an app contact the server? (encryption). 

While security implications are relevant, our current effort focuses on implications related to confidentiality and integrity of data.

\paragraph{Confidentiality} \emph{An app violates confidentiality if it releases data to an untrusted sink.} Hence, an app (and, consequently, its storyboard) that violates confidentiality is deemed as insecure. We explain the concepts of untrusted sinks in a storyboard in \Fref{sec:analysis}.  

\paragraph{Integrity} \emph{An app violates integrity if it uses a value from an untrusted source.} Hence, an app (and, consequently, its storyboard) that violates integrity is deemed as insecure.  We explain how a source is identified as untrusted in \Fref{sec:analysis}.

\subsection{Storyboard Analysis}
\label{sec:analysis}
There are multiple ways to check if extended storyboards satisfy various properties.  In our current realization of SeMA, we use \emph{information flow analysis} and \emph{rule checking} to check and help ensure extended storyboards satisfy security properties concerning confidentiality and integrity. 

\subsubsection{Information Flow Analysis (IF)}
\label{sec:ifa}
This analysis tracks the flow of \textit{information} in the form of values from sources to sinks in a storyboard.

A \emph{source} is either a widget in a screen, an input parameter of a screen, or (the return value of) an operation.  The set of sources in a storyboard are partitioned trusted and untrusted sources based on the guarantee of data integrity.  Specifically, a source is \emph{untrusted} if it is an input parameter of a screen's URI or it is an operation that reads data from an HTTP server, an open socket, device's external storage, or device's clipboard; or uses the capability of a resource provided by another app.  All other sources are deemed as trusted. 

A \textit{sink} is either a widget in a screen or an argument to a screen or an operation.  The set of sinks in a storyboard are partitioned trusted and untrusted sinks based on the guarantee of data confidentiality.  Specifically, a sink is deemed as \textit{untrusted} if it is an argument to an external screen identified without the \textit{app} attribute or to an operation that writes data to an HTTP server, an open socket, device's external storage, or device's clipboard; or uses the capability of a resource provided by another app. All other sinks are deemed as trusted.

The analysis relies on other actors (e.g., developer) and systems (e.g., user authentication and key management by Android runtime platform and app validation and authentication by app stores) to ensure aspects such as authenticity and trust of sources and sinks.

To reason about the flow of information between sources and sinks, we define a binary reflexive relation named \textit{influences} between sources and sinks as follows: \textit{influences(x,y)} (\ie source $x$ influences sink $y$) if 

\begin{enumerate}
  \item $x$ is assigned to an input parameter $y$ of screen $s$ on an incoming transition to $s$,
  \item $x$ is an argument to operation $y$, or
  \item value of operation $x$ or input parameter $x$ of screen $s$ is assigned to widget $y$ in screen $s$.
\end{enumerate}

Here are few instances of this relation in \Fref{fig:story_ex}. \textit{influences(y, Phone)} because input parameter \textit{y} of \texttt{Contacts} screen is assigned to \texttt{Phone} widget.  \textit{influences(x,dispMsg)} because \textit{x} is provided as an argument to operation \texttt{dispMsg}.  \textit{influences(dispMsg, Status)} because the value of operation \texttt{dispMsg} is assigned to \texttt{Status} widget in \texttt{SaveStatus} screen.

Further, all data flows inside the app are guaranteed to preserve the confidentiality and integrity of used/processed data.  The operations provided as part of the capabilities of custom resources provided by the app are assumed to preserve the confidentiality and integrity of used/processed data.

With the above (direct) influence relation, guarantees, and assumptions, we use the transitive closure of \textit{influences} relation to detect violation of security properties.  We detect \emph{potential violation of integrity} by identifying transitive (indirect) influences $(x,y) \in \mbox{\emph{influences}}^*$ in which $x$ is an \emph{untrusted} source.  Likewise, we detect \emph{potential violation of confidentiality} by identifying transitive (indirect) influences $(x,y) \in \mbox{\emph{influences}}^*$ in which $y$ is an \emph{untrusted} sink.  All such identified indirect influences are reported to the developer and must be eliminated.  Such an indirect influence can be eliminated by either replacing untrusted sources/sinks with trusted sources/sinks or indicating the indirect influence as safe by marking one or more of the direct influences that contribute to the indirect influence as \emph{safe}.\footnote{This marking is similar to data declassification in traditional information flow analysis.}\footnote{When an indirect influence is enabled via multiple sequences of direct influences between a source and a sink (similar to multiple paths between nodes in a graph), at least one direct influence in each sequence should be marked as safe to indicate the indirect influence is safe.}

For example, in \Fref{fig:story_ex}, the \textit{y} parameter of \texttt{Contacts} screen is untrusted as it is provided by an external app.  So, the analysis will flag $\mbox{\emph{influences}}^*(y, Phone)$ as violating integrity due to the assignment of \textit{y} to \texttt{Phone} widget in \texttt{Contacts} screen. A developer can fix this violation either by removing the assignment of \textit{y} to \texttt{Phone} or marking the assignment as \textit{safe} (as done in \Fref{fig:story_ex}).

\paragraph{Correctness of the analysis}
The purpose of the analysis is to identify violations of confidentiality or integrity.  This is done in two steps: 1) calculate the potential flows in the storyboard using the transitive closure of \textit{influences} relation and 2) check if a flow involves an untrusted source or sink and is not marked as \emph{safe}.  Since the second step is based on pre-defined classification of sources and sinks and developer-provided \emph{safe} annotations, the correctness of the analysis hinges on the first step.

If the \textit{influences} relation correctly captures the direct flow between sources and sinks in a storyboard, then the transitive closure $\mathit{influences}^*$ will capture all possible flows between the sources and sinks, including all flows violating confidentiality or integrity. Hence, \emph{the analysis is complete} in identifying every violation of confidentiality or integrity.

The $\mathit{influences}^*$ relation does not consider the effect of constraints on flows between sources and sinks; that is, all constraints on transitions are assumed to be true.  Consequently, the analysis may identify a flow between a source and a sink when there is no flow between the source and the sink (at runtime).  For example, suppose a screen $s$ has two incoming transitions $i_1$ and $i_2$ and two outgoing transitions $o_1$ and $o_2$ along with transition constraints that dictate transition $o_x$ must be taken if and only if $s$ was reached via transition $i_x$.  Further, suppose an input parameter $m$ of $s$ is assigned a value along $i_1$ and $i_2$ and used as an argument along $o_1$ and $o_2$. In this case, the value assigned to $m$ along $i_1$ ($i_2$) will not flow out of $m$ along $o_2$ ($o_1$).  However, the analysis will incorrectly identify the value assigned to $m$ along $i_1$ ($i_2$) may flow out of $m$ along $o_2$ ($o_1$). Since the analysis may report invalid violations, \textit{the analysis is unsound}. However, a developer can override such violations by marking relevant direct flows as \emph{safe}.  

A formal description of the information flow analysis along with the algorithm and its correctness proof is provided in Appendix \ref{sec:infoFlowFormal}.

\subsubsection{Rule Checking (RC)}
\label{sec:rule-based}
Prior research has developed guidelines and best practices for secure Android app development ~\cite{Owasp:URL, Standards:URL}. Based on these standards, we have developed rules that can be enforced at design time (in addition to information flow analysis) to prevent the violation of properties related to confidentiality and integrity. 

Following is the list of rules supported by the current realization of SeMA, along with the reasons for the rules.

\begin{enumerate}
  \item \textit{Capabilities offered by custom resources must be protected by access control.}  If any external client can access a custom resource (\ie its \emph{access} attribute is set to \emph{all}) and the resource offers privileged capabilities, then malicious clients can gain access to privileged capabilities without the user's consent.  Further, \emph{Android's policy of least privilege} stipulates that apps should have minimal privileges and acquire the privileges required to use protected services.  
    
  \item \textit{WebView widgets must be configured with a whitelist of URL patterns}. A WebView widget in an app works like a browser -- it accepts a URL and loads its contents -- but it does not have many of the security features of full-blown browsers.  Also, a WebView widget has the same privileges as the containing app, has access to the app's resources, can be configured to execute JavaScript code.  Hence, loading content from untrusted sources into WebView widgets facilitates exploitation by malicious content.
    
  \item \textit{Operations configured to use HTTPS remote servers must use certificate pinning.}  HTTPS remote servers are signed with digital certificates issued by certificate authorities (CAs). Android defines a list of trusted CAs and verifies that the certificate of an HTTPS remote server is signed with a signature from a trusted CA. However, if a trusted CA is compromised, then it can be used to issue certificates for malicious servers. Hence, to protect against such situations, certificates of trusted servers are pinned (stored) in apps, and only servers with these certificates are recognized as legit servers by the apps.
    
  \item \textit{Operations configured to use SSL sockets must use certificate pinning.} The reasons from the case of certificate pinning for HTTPS applies here as well.
  
  \item \textit{Cipher operations must use keys stored in secure key stores (containers)}. The results of cipher operations can be influenced by tampering the cryptographic keys used in cipher operations.  Further, since cryptographic keys are often used across multiple executions of an app, they need to be stored in secondary storage that is often accessible by all apps on a device. Hence, to protect against unwanted influences via key tampering, cipher keys should be stored in secure key stores (containers).
    
\end{enumerate}

\paragraph{Realization of Rule Checking}

Violations of rule 1 are detected by checking if a custom resource offers a privileged capability and has its \emph{access} attribute set to \emph{all}.

The \textit{trust-patterns} attribute of \textit{WebView} widget is used to specify the whitelist of trusted URL patterns.  Violations of rule 2 is detected by checking if \textit{trust-patterns} attribute is specified for every \textit{WebView} widget.

Violations of rule 5 are detected by checking if the key argument provided to a cipher operation is the value returned by a pre-defined operation to keys from a secure container.

Violations of rules 1, 2, and 5 are flagged as errors and must be addressed before moving to the code generation phase.  

Certificate pinning is enabled by default in every storyboard in the methodology.  However, since techniques other than certificate pinning can be used to secure connections to servers, a developer can disable certificate pinning by setting \textit{disableCertPin} attribute in a network-related operation.  Such cases are detected as violations of rules 3 and 4.  They are flagged as warnings but do not inhibit the developer from moving to the code generation phase. 

\subsection{Code Generation}
\label{sec:codegen}

\begin{figure*}
    \centering
    \includegraphics[width=\linewidth]{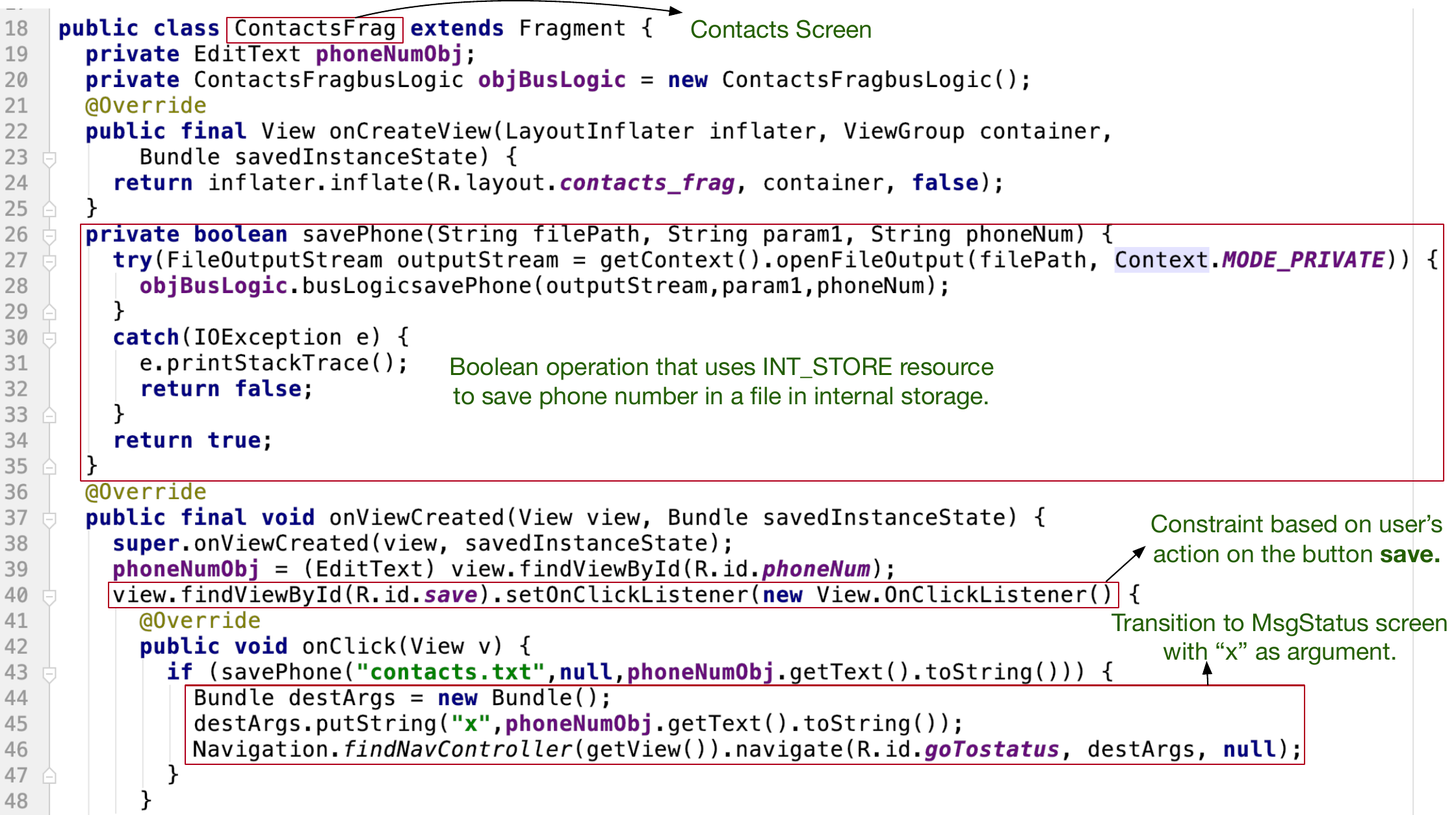}
    \caption{Code generated for the Contacts screen in the storyboard depicted in \Fref{fig:story_ex}}
    \label{fig:codeGen_ex}
\end{figure*}

Once the developer has verified that the specified storyboard does not violate properties related to confidentiality and integrity, she can generate code from the storyboard. \Fref{fig:codeGen_ex} shows a fragment of generated code for the running example. 

The current realization of SeMA hinges on various choices in mapping and translating storyboard-level entities and concepts into code-level entities and concepts. These choices are encoded in the following mapping and translation rules used during code generation.

\begin{enumerate}
  \item A Screen is translated to a \texttt{Fragment}. For each input parameter of the screen, a function to obtain the value of the parameter is generated.  If an input parameter is not available at runtime, then the corresponding function raises a runtime exception.
    
  \item A widget is translated to the corresponding widget type in Android, \eg a widget displaying text is translated to \texttt{TextView}. The value of the widget is the corresponding value specified in the storyboard. For example, if the value is provided by a screen's input parameter \textit{x}, then the return value from the getter function of \textit{x} is set as the widget's value, \eg \texttt{TextView.setText(getX())}. The value in a widget is obtained via the corresponding getter function, \eg \texttt{TextView.getText()}.
    
  \item The constraint associated with a transition from a source screen to a destination screen is a conjunction of a user action and boolean operations. The user action part of the constraint is translated to a listener/handler function in the source screen that is triggered by the corresponding user action, \eg button click.  If the constraint has a boolean expression, then a conditional statement is generated with the boolean expression as the condition in the body of the listener function. The \emph{then} block of the conditional statement has the statements required to trigger the destination screen. If the constraint has no boolean operations, then the body of the listener function has statements required to trigger the destination screen.  If the constraint has no user action, then the checks corresponding to the boolean operations are performed when the source fragment is loaded.
  
  We use Android's navigation APIs to trigger a destination screen. If the destination screen is a proxy screen, then intents are used to trigger the destination screen determined by the \textit{URI} and \textit{app} attribute. Arguments to destination screens are provided as key/value pairs bundled via the \texttt{Bundle} API.
  
  When a screen has multiple outgoing transitions, the statements corresponding to the transitions are chained in the specified order of the transitions in the storyboard.
    
  \item An operation is translated to a function with appropriate input parameters and return value. Each reference to the operation in a storyboard is translated to call the corresponding function. 
  
  The type of the input parameters depends on the type of the arguments provided to the function. For example, if the argument is provided by a widget that displays text, then the type of the parameter will be \texttt{String}. The return type depends on how the function is used. For example, if the function is used as a boolean operation in a constraint, then its return type will be \texttt{boolean}. If the function is assigned to a widget that displays text, then the function's return type will be \texttt{String}. 
    
  If the operation uses a capability provided by a pre-defined resource, then the body of the corresponding function will contain the statements required to use the capability. Otherwise, the function will have an empty body that needs to be later filled in by the developer. For example, on line 27 in \Fref{fig:codeGen_ex}, function \textit{savePhone} contains the statement required to create a file in the device's internal storage since the same operation uses \textit{write} capability of the resource \textit{INT\_STORE} in the storyboard.
  
\item A developer can extend the generated definitions of functions. For example, on line 28 in \Fref{fig:codeGen_ex}, the generated code provides a hook for the developer to extend \textit{savePhone}.
  
  \item A custom resource is translated to an appropriate Android component.  The capabilities provided by a custom resource can be accessed via an Android intent.  Currently, we only support broadcast receivers as custom resources. 

  \item The use of a resource in the storyboard indicates an app depends on the resource.  Such dependencies are captured in the app's configuration during code generation while relying on the Android system to satisfy these dependencies at runtime in accordance with the device's security policy, \eg grant permission to use a resource at install time.
  
%   \item A developer can set a flag in the build process to indicate that code generation should not remove developer-added code after updating the storyboard. While setting this flag will not break code compilation, it could result in dead code since generated code references developer-added code but not vice-versa.     

\end{enumerate}

\section{Implementation of the Methodology}
Android JetPack Navigation (AJN) is a suite of libraries that helps Android developers design their apps' navigation in the form of navigation graphs. A navigation graph is a realization of a traditional storyboard in Android Studio. We have extended navigation graphs with features that enable developers to specify an app's storyboard in Android Studio. The developer can visually represent the screens, widgets, and transitions in the navigation graph. While a developer cannot specify operations and constraints visually, she can specify them in the corresponding XML structure of the navigation graph.

We have extended Android Lint ~\cite{Lint:URL}, a static analysis tool to analyze files in Android Studio, to implement the analysis and verification of security properties. The analysis is packaged as a Gradle Plugin that can be used from Android Studio.

We have implemented a code generation tool that takes a navigation graph and translates it into Java code for Android. A developer can extend the generated code with business logic in Java or Kotlin. The code generation tool is also packaged as a Gradle Plugin that can be used from Android Studio.

\paragraph{Designing an app in Android Studio with SeMA} An app developer uses the existing AJN libraries to build a navigation graph of the app.  This graph serves as the initial storyboard (\eg \Fref{fig:init}).  Next, she uses the SeMA provided extensions to the navigation graph to iteratively specify the app's behavior. In each iteration, she uses the extensions to Android Lint tool to analyze the navigation graph for violations of pre-defined security properties. Once the developer has verified the storyboard satisfies the desired properties, she uses the code generation tool to generate an implementation of the app. Finally, she adds the business logic to the generated implementation to complete the implementation of the app.  

All through the development process, the developer operates within the existing development environment while exercising the extensions provided by SeMA.

\section{Canonical Examples}

\begin{figure}
    \centering
    \includegraphics[width=\linewidth]{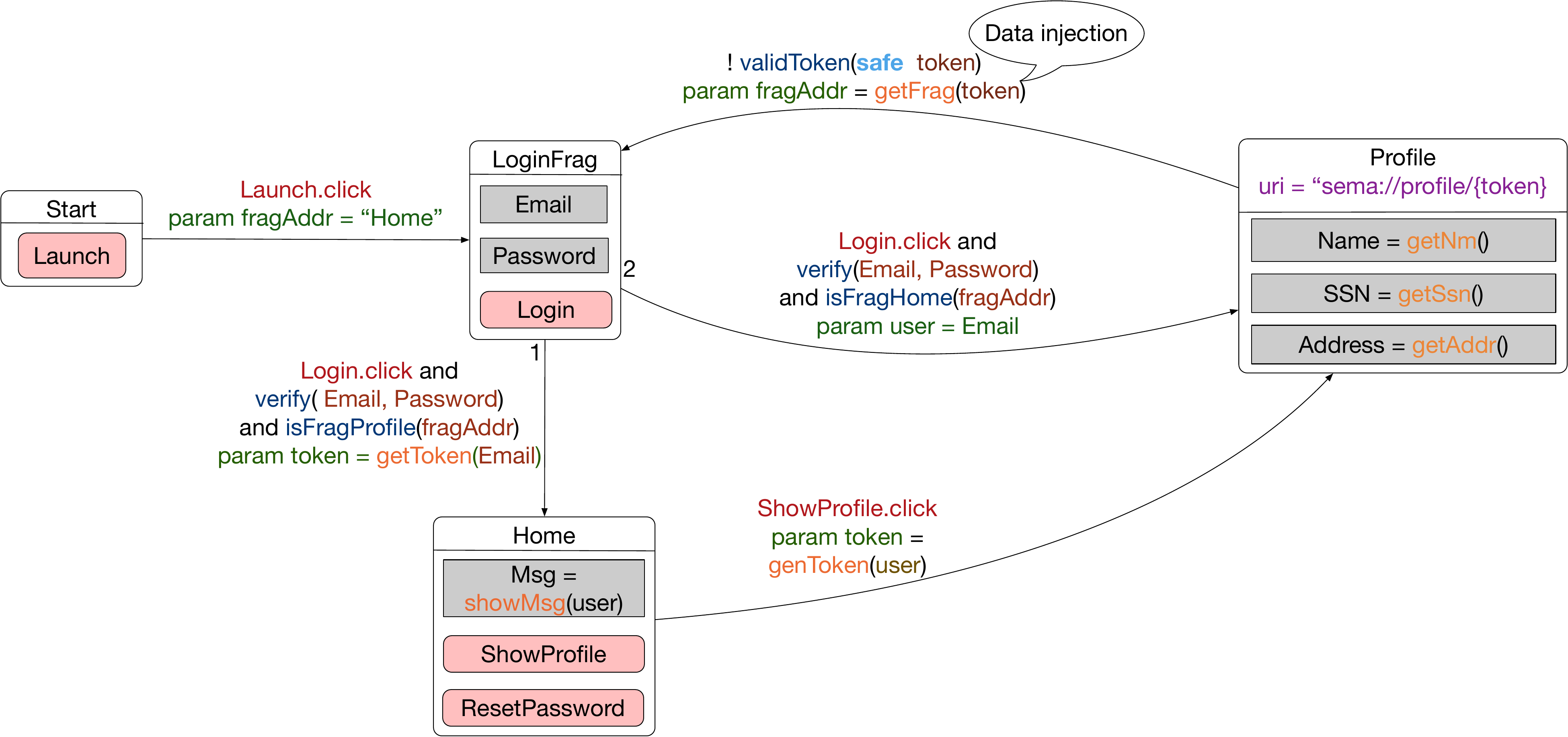}
    \caption{Example illustrating a data injection vulnerability. The (information) bubbles are not part of the storyboard.}
    \label{fig:type-ex}
\end{figure}

We illustrate the methodology with a couple of canonical examples based on vulnerabilities found in real-world apps. Each example demonstrates a vulnerability that can be detected and prevented by the methodology at design time. For each example, we describe the expected behavior and outline the steps a developer would take to specify the app's behavior in a storyboard and verify the security properties related to confidentiality and integrity. The steps outlined are intended to prescribe one way of developing an app with SeMA. SeMA does not enforce a workflow. Developers may choose to build a storyboard iteratively or they may choose to make it all at once. However, SeMA enforces developers to verify the specified storyboard for security violations before translating it to code.

\subsection{Example Illustrating the Use of Information Flow Analysis}
Consider an app that allows users to log in and view their profile information. From another app, a valid user of the app can navigate \textit{only} to the screen showing profile information. The example is based on the fragment injection vulnerability discovered in real-world apps ~\cite{IBMSec:URL}.

A developer will specify an app in SeMA as follows:

\paragraph{Specify screens and navigation} A developer starts by specifying the screens of the app, the widgets in each screen, and the possible transitions between the screens. For example, in \Fref{fig:type-ex}, a developer initially specifies 4 screens: \texttt{Start}, \texttt{LoginFrag}, \texttt{Home}, and \texttt{Profile}. Each screen has widgets, \eg the \texttt{Start} screen has one button with label \texttt{Launch}. Finally, the screens are connected to each other via transitions to indicate how the user can navigate between screens, \eg from the \texttt{LoginFrag}, a user can navigate to either \texttt{Home} or \texttt{Profile}. 

\paragraph{Add user action-related constraints to transitions} A developer adds constraints to transitions that correspond to actions/gestures performed by users on the widgets in the screens. In \Fref{fig:type-ex}, such constraints are highlighted in red. For example, the transition from \texttt{Start} to \texttt{LoginFrag} is taken when \texttt{Launch} button in \texttt{Start} is clicked. 

\paragraph{Refine transitions with additional constraints} A developer adorns a transition with guards/constraints not related to user actions. Such constraints are highlighted in blue in \Fref{fig:type-ex}. For example, the transition from \texttt{LoginFrag} to \texttt{Home} is taken when the \texttt{verify} operation and the \texttt{isFragHome} operation returns true in addition to the user action-related constraint, that is, \texttt{Login} button clicked. Likewise, the transition from \texttt{Profile} to \texttt{LoginFrag} is taken when the \texttt{validToken} operation returns true. 

These boolean operations may have input parameters. For example, the \texttt{verify} operation takes two input parameters, the arguments to which are provided by the values in widgets \texttt{Email} and \texttt{Password}.

\paragraph{Connect input parameters of screens to data sources} A developer specifies the data sources that will provide arguments to the input parameters of a screen. These sources can be provided in two ways -- (1) as part of incoming transitions to a screen (in green \Fref{fig:type-ex}) or (2) as part of URIs associated with screens (in purple in \Fref{fig:type-ex}). For example, the argument to the \texttt{token} input parameter of screen \texttt{Profile} is provided by a non-boolean operation \texttt{getToken} (highlighted in orange in \Fref{fig:type-ex}) when the transition to \texttt{Profile} is taken. On the other hand, the argument of \texttt{token} is provided by the \texttt{token} variable associated with the URI of screen \texttt{Profile} when \texttt{Profile} is triggered by an external app via the associated URI.

\paragraph{Check security properties} Finally, the developer analyzes the specified storyboard for security property violations, which results in a data injection warning, as shown in \Fref{fig:type-ex}. This warning is because \texttt{isFragProfile} and \texttt{isFragHome} operations consume the input parameter \texttt{fragAddr} as argument.  On the transition from \texttt{Profile} screen to \texttt{Login Frag} screen, \texttt{fragAddr} takes on the return value of \texttt{getFrag} operation that consumes \texttt{token} parameter of \texttt{Profile} screen.  Since an external app provides the \texttt{token} argument to the \texttt{Profile} screen via its URI, an external app can manipulate \texttt{token} to gain access to \texttt{Home} screen. Information flow analysis will detect and flag this violation by following the chain of flow from untrusted source and sinks.

\paragraph{Apply fix suggested by SeMA} The reported vulnerability can be fixed by changing \texttt{param fragAddr = getFrag(token)} to \texttt{param fragAddr = "profile"} on the transition from \texttt{Profile} screen to \texttt{LoginFrag} screen as this breaks the dependence between the app's navigation and \texttt{token} parameter.

\begin{figure}
    \centering
    \includegraphics[width=\linewidth]{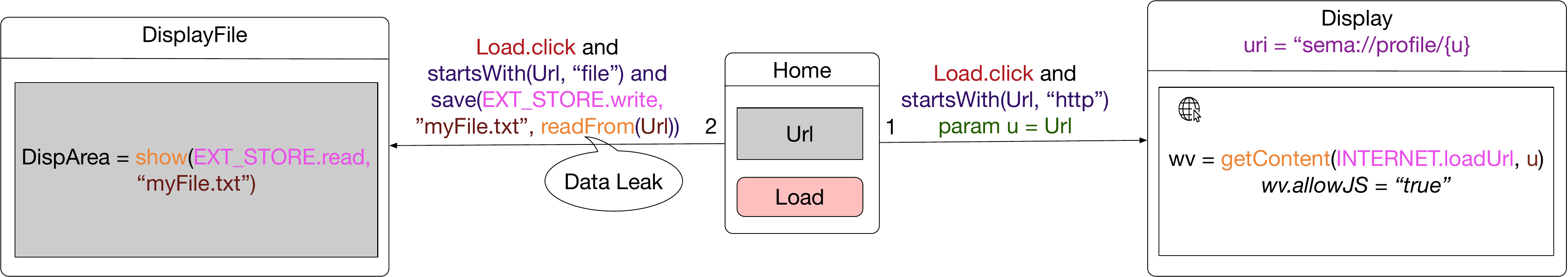}
    \caption{Example illustrating a data leak vulnerability. The (information) bubbles are not part of the storyboard.}
    \label{fig:dl-ex}
\end{figure}

\subsection{Example Illustrating the Use of Information Flow Analysis and Rule Checking}
Consider an app that allows a user to enter a trusted URL in a text field and displays the content from the URL. If the URL is a file URL, then the data in the file is displayed to the user and downloaded to the device's external storage. The URL can be a file URL or a web URL. The example is based on the sensitive data leak vulnerability discovered in the Firefox app for Android ~\cite{Mozilla:URL} and the Zomato app ~\cite{Zomato:URL}.  

A developer will specify the storyboard of this app iteratively as follows: 

\paragraph{Specify screens and navigation} As in the previous example, a developer starts by specifying the screens, widgets in each screen, and the possible transitions between the screens. As shown in \Fref{fig:dl-ex}, the app has 3 screens -- \texttt{Home}, \texttt{Display}, and \texttt{DisplayFile}. The \texttt{Display} screen is used to display web content from a URL entered in \texttt{Url} in \texttt{Home} screen. The \texttt{DisplayFile} screen is used to display the contents of a file entered by the user in \texttt{Url} in \texttt{Home} screen.

\paragraph{Add user action-related constraints to transitions} A developer attaches constraints on transitions related to gestures/actions by a user on widgets in a screen (highlighted in red \Fref{fig:dl-ex}). For example, in the figure, a user can either navigate to \texttt{Display} or \texttt{DisplayFile} from \texttt{Home} upon the click of the \texttt{Load} button in \texttt{Home}.

\paragraph{Refine transitions with additional constraints} A developer adds guards/constraints on transitions not related to user actions (shown in blue \Fref{fig:dl-ex}). These guards are specified as operations (in orange in \Fref{fig:dl-ex}) that return boolean values. For example, the transitions from \texttt{Home} to \texttt{Display} is taken when the \texttt{Load} button is clicked and the \texttt{startsWith} operation returns true. The \texttt{startsWith} verifies if the URL entered by the user starts with "http:". Likewise, the transition from \texttt{Home} to \texttt{DisplayFile} is taken when the \texttt{Load} button is clicked, the entered URL starts with "file" (see \texttt{startsWith} operation in \Fref{fig:dl-ex}), and the \texttt{save} operation returns true (\ie the content in the file path \texttt{Url} is saved). 

\paragraph{Specify pre-defined resources used by operations} A developer adds more detail to the operations via pre-defined resources. For example, a developer specifies the type of local storage that will be used by the \texttt{save} operation, as shown in \Fref{fig:dl-ex} (highlighted in pink). In the storyboard shown, the \texttt{save} operation uses the \texttt{write} capability provided by the pre-defined resource \texttt{EXT\_STORE} to write data to a file in the device's external storage.

\paragraph{Connect input parameters of screens to data sources} As in the previous example, a developer specifies a data source to \texttt{u}, an input parameter of screen \texttt{Display} in \Fref{fig:dl-ex}. The source is either the widget \texttt{Url} in screen \texttt{Home} (in green in \Fref{fig:dl-ex}) or the variable \texttt{u} in the URI associated with the \texttt{Display} screen (in purple).

\paragraph{Connect widgets to data sources} A developer specifies the data source that provides the value to be displayed in a widget. As shown in \Fref{fig:dl-ex}, the values from non-boolean operations \texttt{show} and \texttt{getContent} are used as values to be displayed in \texttt{DispArea} and \texttt{wv} widgets respectively. The operation \texttt{show} reads a file in the device's external storage and the operation \texttt{getContent} loads web content from a URL. 

\paragraph{Configure WebView widgets} A developer configures the \texttt{Webview} widget by setting pre-defined attributes of the widget. For example, in \Fref{fig:dl-ex} JavaScript execution is allowed by setting the value of the widget's \texttt{allowJS} attribute to "true".

\paragraph{Check security properties} Finally, a developer analyzes the specified behavior in the storyboard to check for violations related to confidentiality and integrity. The analysis reveals a violation of confidentiality due to a data leak vulnerability. If the file path provided in \texttt{Url} is a file in the app's internal storage, then any file from the app's internal storage can be stored in the device's external storage. Thus, all apps installed in the device can access the app's internal files since files in external storage can be accessed by all apps. This violates the confidentiality of the data read from trusted sources. Information flow analysis will detect and flag this violation by following the chain of flow from sources to untrusted sinks. Further, rule checking will flag a violation of Rule 2 because a whitelist of trusted URLs is not specified for \texttt{wv} widget via \texttt{trust-patterns} attribute. The absence of whitelist will allow \texttt{wv} to execute JavaScript embedded in potentially malicious URLs.

\paragraph{Apply fix suggested by SeMA} One way to fix this vulnerability is not to save internal files to external storage. Instead, the data in an internal file should be directly displayed in \texttt{DisplayFile}. The user should be provided with a button, in \texttt{DisplayFile}, to download the displayed data in external storage. This design is more secure than the one described above because data in the app's internal files cannot be exposed without the user's explicit approval via a UI action, \eg button click. However, SeMA will still report a violation since data is being written to an external storage file. In this case, since the app's design requires explicit consent from the user before writing to external storage, a developer can override this violation by marking the flow of information to the external storage file as \textit{safe}. 

The rule violation is fixed by specifying a whitelist of trusted URLs for \texttt{wv}, \eg \texttt{wv.trust-patterns=\linebreak\{".*sema.org.*"\}}.

\section{Evaluation}
\label{sec:eval}
We evaluated SeMA in terms of (1) \textit{Effectiveness} focused on its ability to detect and prevent vulnerabilities, and (2) \textit{Usability} focused on its ability to help developers build apps with features that are often found in real-world apps but without a set of known Android app vulnerabilities.

\subsection{Effectiveness}
\label{sec:evalFeas}
We used the Ghera benchmark suite~\cite{Mitra:2017} for this evaluation. Ghera has 60 benchmarks. Each benchmark captures a unique vulnerability. We used Ghera because the vulnerabilities in Ghera benchmarks are \textit{valid}; that is, they have been previously reported in the literature or documented in Android documentation. These vulnerabilities are \textit{general and exploitable} as they can be verified by executing the corresponding benchmarks on vanilla Android devices and emulators. Further, each vulnerability is \textit{current} as they are based on Android API levels 22 thru 27, which enable more than 90\% of Android devices in the world and are targeted by both existing and new apps. Finally, the benchmarks are \textit{representative} of real-world apps in terms of the APIs they use, as established by Ranganath and Mitra ~\cite{Ranganath:EMSE19}. Hence, the benchmarks in Ghera are well-suited for this evaluation.

For each Ghera benchmark, we used SeMA to create a storyboard of the \textit{Benign} app. If SeMA detected the vulnerability in \textit{Benign} app's storyboard, then we modified the storyboard till SeMA found no vulnerabilities. We generated the code from this storyboard and verified the absence of the vulnerability captured in \textit{Benign} by executing it with the \textit{Malicious} app. If the \textit{Malicious} app was unable to exploit the generated app, then we concluded that SeMA prevented the vulnerability. 

Of the 60 vulnerabilities captured in Ghera benchmarks, SeMA detected 49 vulnerabilities: 15 were detected via information flow analysis, nine were detected via rule checking, six were detected via a combination of info flow analysis and rule checking, and 19 were prevented via code generation. \Fref{tab:results} provides the breakdown of this information.

\begin{longtable}{llr}
\caption{Results showing how a vulnerability in a benchmark was detected/prevented. CG, IF, and RC refer to Code Generation, Information Flow Analysis, and Rule-based Analysis, respectively.}\\
\hline
\textbf{ID} & \textbf{Benchmark} & \textbf{Method}\\
\hline
\endfirsthead
\multicolumn{3}{c}%
{\tablename\ \thetable\ -- \textit{Continued from previous page}} \\
\hline
\textbf{ID} & \textbf{Benchmark} & \textbf{Method}\\\\
\hline
\endhead
\hline \multicolumn{3}{r}{\textit{Continued on next page}} \\
\endfoot
\hline
\endlastfoot
     C1 & Block Cipher in ECB mode weakens the encryption. & CG\\
     C2 & Non-random IV weakens the encryption. & CG \\
     C3 & Use of a constant secret key exposes the secret key. & RC\\
     C4 & No authentication of KeyStore exposes the KeyStore. & CG\\
     C5 & Constant Salt in password-based secret key exposes the key. & CG\\
    %  \multicolumn{3}{c}{Inter-Component Communication} \\
    %  \hline
     I1 & Dynamic registration of broadcast receiver exposes it. & RC\\
     I2 & Sharing empty pending intent leads to intent hijacking. & IF\\
     I3 & Dynamic Fragment loading enables fragment injection. & IF\\
     I4 & Low priority activity can be hijacked by high priority activity. & IF \\
     I5 & Sharing implicit pending intent leads to intent hijacking. & IF \\
     I7 & Not validating implicit intents grants unauthorized access. & IF\\
     I8 & Not verifying a broadcast message exposes the receiver. & RC\\
     I9 & Accepting input from a receiver enables data injection. & IF\\
     I10 & Not protecting broadcast receiver enables privilege escalation. & RC\\
     I11 & Starting activity in a new task enables activity hijack. & CG\\
     I12 & Non-empty task affinity of launcher activity enables phishing. & CG\\
     I13 & Starting activity in a new task enables phishing. & CG\\
     I14 & Non-empty task affinity of an activity enables phishing and DoS. & CG\\
    %  \multicolumn{3}{c}{Networking} \\
    %  \hline
     N1 & Incorrect checking of server certificates enables MITM. & CG\\
     N2 & Incorrect verification of server hostname enables MITM. & CG \\
     N3 & Not verifying the hostname of an SSL socket enables MITM. & CG\\
     N4 & \textit{SSLCertificateSocketFactory.getInsecure()} API enables MITM. & CG \\
     N5 & Incorrect checking of certificate authority enables MITM. & CG \\
     N6 & Sending data via open socket enables data theft. & IF\\
     N7 & Reading data from an open socket enables data injection. & IF\\
     N8 & Absence of pinned certificates enables MITM. & RC\\
    %  \multicolumn{3}{c}{Permission} \\
    %  \hline
     P1 & Using unnecessary permissions enables privilege escalation. & CG\\
     P2 & Weak permissions do not provide adequate protection. & RC\\
    %  \multicolumn{3}{c}{Storage}\\
    %  \hline
     S1 & Reading from files in external storage enables data injection. & IF\\
     S2 & Writing to files in external storage enables data theft. & IF\\
     S3 & Unsanitized paths to internal storage exposes internal files. & IF\\
     S4 & Writing to files in external storage from files in internal storage enables data theft. & IF\\
    %  \multicolumn{3}{c}{System}\\
    %  \hline
     Y1 & API \textit{CheckCallingOrSelfPermission} does not provide protection. & CG\\
     Y2 & API \textit{CheckPermission} does not provide protection. & CG\\
     Y3 & API \textit{EnforceCallingOrSelfPermission} does not provide protection. & CG\\
     Y4 & API \textit{EnforceCallingOrSelfPermission} does not provide protection. & CG\\
     Y5 & Writing data to the clip board enables data theft. & IF\\
     Y6 & Executing code from unverified source enables code onjection. & IF\\
     Y7 & Sharing unique system ID enable ID theft. & IF\\
    %  \multicolumn{3}{c}{Web}\\
    %  \hline
     W1 & Malicious URLs in a WebView overwrite cookies. & IF \& RC\\
     W2 & Connecting to HTTP remote servers enables MITM. & RC\\
     W3 & Malicious URL in a WebView can inject \& execute JS. & IF \& RC\\
     W4 & Malicious URL in a WebView can embed intents to steal data. & IF \& RC \\
     W5 & Malicious URL in a WebView can access internal content providers. & IF \& RC\\
     W6 & Malicious URL in a WebView can access internal files. & IF \& RC\\
     W7 & Ignoring SSL errors enables MITM. & CG\\
     W8 & A WebView not validating resource requests enables MITM. & RC\\
     W9 & A WebView with no base URL enables file access to a malicious URL. & IF \& RC\\
     W10 & A WebView not validating page requests enables MITM. & RC\\
     \label{tab:results}
\end{longtable}

Two of the vulnerabilities detected via rule analysis could have been detected by code generation. These two vulnerabilities relate to connecting to HTTP remote servers and connecting to HTTPS remote servers without certificate pinning. While \textit{such vulnerabilities can be prevented by code generation, our realization of the methodology relies on rule checking to offer flexibility in using HTTP vs. HTTPS and certificate pinning in storyboards when connecting to remote servers}.

The current realization of SeMA was not applicable to 11 benchmarks in Ghera. Of these, three benchmarks capture vulnerabilities that cannot be prevented by the methodology, \eg unhandled exceptions. The remaining eight apps use features that are not yet supported by the methodology, \eg Content Providers. 

\subsubsection{Comparison with Existing Efforts} 

Of the 49 vulnerabilities prevented by the methodology, 28 can be detected curatively by source code analysis after implementing the apps. \textit{Detecting the remaining 21 vulnerabilities by source code analysis is harder} due to combinations of factors such as the semantics of general-purpose programming languages (\eg Java), security-related specifications provided by the developer (\eg source/sink APIs), and the behavior of the underlying system (\eg Android libraries and runtime). This observation is supported by Ranganth and Mitra ~\cite{Ranganath:EMSE19}, who found that existing source code analysis tools are not effective in detecting vulnerabilities in an earlier version of Ghera which included 15 of these 21 vulnerabilities. Further, Pauck \etal ~\cite{Pauck:2018} evaluated six prominent static taint analysis tools aimed to detect data leak vulnerabilities in Android apps and discovered that most tools detect approximately 60\% of the vulnerabilities captured in the DroidBench 3.0 benchmark suite. Finally, Luo \etal ~\cite{Luo:2019} qualitatively analyzed Android app static taint analysis tools and observed that these tools need to be carefully configured (\eg relevant source/sink APIs) and should consider application context to accurately detect vulnerabilities in Android apps.

% Of the 49 vulnerabilities prevented by the methodology, 28 can be detected curatively (by source code analysis) after implementing the apps. Detecting the remaining 21 vulnerabilities by source code analysis is harder due to combinations of factors such as the semantics of general-purpose programming languages (\eg Java), security-related specs provided by the developer (\eg source/sink APIs), and the behavior of the underlying system (\eg Android libraries and runtime). Hence, \textit{many existing tools based on source code analysis fail to detect vulnerabilities in real-world apps effectively} ~\cite{Ranganath:EMSE19,Pauck:2018,Luo:2019}.

Ranganath and Mitra ~\cite{Ranganath:EMSE19} found that 14 security analysis tools in isolation could detect at most 15 vulnerabilities, and the full set of tools collectively detected 30 vulnerabilities.  This result suggests combining different analysis (\eg, deep analyses like pointer analysis, shallow analyses like pattern matching) will likely be more effective in detecting vulnerabilities. Our experience with SeMA suggests \textit{the same is likely true in the context of preventing vulnerabilities: a combination of information flow analysis (deep analysis), rule checking (shallow analysis), and code generation (shallow analysis) helped detect and prevent 49 vulnerabilities}.

Gadient \etal ~\cite{Gadient:2020} found that real-world apps often expose credentials (\eg crypto keys) in the source code, use insecure communication channels (\eg HTTP), and use malicious input to load URLs in a WebView. These vulnerabilities are also captured in Ghera and were prevented by SeMA in this evaluation. This finding suggests \textit{there are non-trivial opportunities for techniques like SeMA to help improve security of real-world apps}.

\subsection{Usability}
\label{sec:evalUsability}
While the evaluation with the Ghera benchmarks shows that SeMA can be used to prevent known vulnerabilities in small apps, it does not tell us if SeMA can be used by developers to uncover vulnerabilities in apps with real-world capabilities and features. Further, SeMA extends storyboards, an existing design artifact, to enable formal reasoning of security properties at design time. While formal reasoning approaches have been proven to be effective in terms of uncovering defects in software (\eg bugs and vulnerabilities), they can be a burden on developers in terms of time to learn and use ~\cite{Schaffer:2016}. Consequently, the adoption of such approaches in domains like mobile app development may be limited. 

To address the concerns identified above, we conducted a usability study of SeMA with ten developers and 13 real-world apps. In this study, we will answer the following research questions:

\begin{itemize}
    \item \textbf{RQ1}: Does SeMA affect app development time?
    \item \textbf{RQ2}: Does software development experience affect development time while using SeMA?
    \item \textbf{RQ3}: Does security-related feature development experience affect development time while using SeMA?
    \item \textbf{RQ4}: Do features used in an app affect development time while using SeMA?
    \item \textbf{RQ5}: Does SeMA detect specific (expected) vulnerabilities introduced in an app's storyboard? 
    \item \textbf{RQ6}: Does the use of SeMA (instead of the usual Android app development process) increase the likelihood of introducing expected vulnerabilities?
    \item \textbf{RQ7}: Does software development experience affect the vulnerabilities introduced while using SeMA?
    \item \textbf{RQ8}: Does security-related feature development experience affect the vulnerabilities introduced while using SeMA?
\end{itemize}

% \begin{enumerate}
%      \item \textbf{RQ1}: \textit{Does SeMA affect app development time?} To answer this question, ten developers developed an Android app \textit{using SeMA}. Next, they developed the same app \textit{without SeMA}. We compared the time taken by a developer to create an app with and without SeMA to understand if SeMA affects Android app development in terms of the time required to make an app. 
     
%     \item \textbf{RQ2 Efficacy}: \textit{Does SeMA help developers prevent vulnerabilities in an Android app?} To answer this question, ten developers developed three apps each, using the SeMA methodology. Each app had vulnerabilities that a developer was expected to introduce while building the app. We calculated the proportion of apps in which developers introduced the expected vulnerabilities and SeMA detected them. This metric was used to measure the effectiveness of SeMA in terms of its ability to detect vulnerabilities and prevent them from propagating into an app's implementation.
    
% \end{enumerate}

In the following sections, we will describe the different aspects of the experiment, explain the various design decisions, and present the results of our analysis on the observed data.

\subsubsection{Study Design}
We designed an experiment to gain insights into the effect of introducing SeMA to the Android app development process. Specifically, we wanted to determine if SeMA helps a developer prevent vulnerabilities when developing an app. Additionally, we wanted to find out the cost of using SeMA in terms of time taken to build an app with SeMA. Hence, we carried out the following tasks to accomplish the study:

\begin{enumerate}
    \item Identify real-world Android apps with expected vulnerabilities.
    \item Hire developers to participate in the study.
    \item Conduct interventions for the developers participating in the study.
    \item Create development exercises for the developers participating in the study.
    \item Assign exercises to developers.
    \item Collect observational data during exercises.
    \item Collect information about developer's prior experience and developer's opinion about SeMA.
    \item Analyze collected data.
\end{enumerate}

In the following paragraphs, we will explain the details and the rationale for performing the above tasks.

\paragraph{Identifying Real-world Apps:} 
For this experiment, we selected 30 Ghera benchmarks.  The selection was guided by the possibility of detecting the captured vulnerabilities via information flow analysis or rule checking.  These vulnerabilities are representative of vulnerabilities found in real-world apps based on API usage information, as shown by Ranganath and Mitra ~\cite{Ranganath:EMSE19}. Hence, they are appropriate for this exercise.

Next, we randomly collected 50 apps from Google Play, Android's official app store. We manually analyzed the source code of the apps to look for features used in the 30 Ghera benchmarks. The first 13 apps we analyzed accounted for the features used by 26 of the 30 selected Ghera benchmarks. None of the remaining 37 apps used features used by the remaining four of the selected Ghera benchmarks. Of the 13 apps, seven apps had been publicly reported to have at least one of the 30 Ghera vulnerabilities ~\cite{HackerOne:URL}. Based on the used features, each of the 13 apps were associated with at least two, on average six, and at most 12 vulnerability benchmarks. The associated number of benchmarks hints at the number of vulnerabilities to expect if these apps were recreated. We refer to these vulnerabilities as expected vulnerabilities from hereon. \Fref{tab:bench-vul-dist} provides a brief description of each app and lists the expected vulnerabilities in each app.

\begin{table*}
\centering
  \ifdef{\TopCaption}{
  \caption{A description of the 13 apps along with the expected Ghera vulnerabilities in each app.}
  }{}
  \begin{tabular}{@{}rlllr@{}} 
    \toprule
    ID & App Name & App Description & Expected Ghera Vuln. & Total\\
    \midrule
    1 & IRS2Go & Registered users can check tax & I4,N8,S1,S2,W2 & 5\\
    & & refund status and make tax-related payments & \\
    2 & Geico & Registered users can view and & N8,S1,S2,Y7,W2 & 5\\
    & & download their auto insurance policies & \\
    3 & Slack & Registered users can chat and call other & I3,I4,I7,N6,N7,N8,W2 & 7\\
    & & registered users & \\
    4 & Ancestry & Registered users can track and save their family & N8,S2,W2 & 3\\
    & & history details & \\
    5 & MyBlockHR & Registered users can upload tax-related & I3,I4,I7,N8,S1,S2,S3,W2 & 8\\
    & & documents and estimate their annual tax & \\
    6 & AESCrypto & Users can encrypt and decrypt & C3,S1 & 2\\
    & & messages with a password & \\
    7 & Grab & Users can book transportation & W1,W3,W4,W6,W8,W9,W10 & 7\\
    & & and accommodation at a location &\\
    8 & Zomato & Users can order food from restaurants & I7,W1,W3,W4,W6,W8, & 8\\
    & & near their location & W9,W10\\
    9 & Clipboard & Users can take notes and share them & N8,Y5,W2 & 3\\
    & Manager & via the clipboard or save to a remote server & \\
    10 & IRCCloud & Users can upload files in their device to the app & I7,S1,S2,S3,S4 & 5\\
    11 & Harvest & Users can upload bills and receipts of their & I3,I4,S1,S2,S3,S4 & 6\\
    & & monthly expenditure &\\
    12 & Firefox & Users can view web or file content in a & S2,S4,W1,W3,W4,W6,W8, & 9\\
    & & custom browser & W9,W10\\
    13 & Yandex & Users can view web content in a custom & I1,I8,I9,I10,P2,W1,W3,W4, & 12\\ 
    & & browser and report a browser crash & W6,W8,W9,W10\\
    \bottomrule  
  \end{tabular}
  \ifundef{\TopCaption}{
    \caption{A description of the 13 apps along with the expected Ghera vulnerabilities in each app.}
  }{}
  \label{tab:bench-vul-dist}
\end{table*}

\paragraph{Hiring Developers:} Since the purpose of this evaluation is to measure the effect of SeMA on Android app development, the obvious candidates for this study are Android app developers. Due to limited local population of developers familiar with Android app development, we reached out to local developers with some professional software development experience. We curated a list of 30 professional developers sourced from personal contacts and sent out personalized emails to each one of them, asking them about their interest in participating in the study. In these emails, we promised to provide a small financial incentive (\eg gift card) if they participated in the study. Of the 30 developers, 15 agreed to participate in the study. While we started the study with 15 participants, five participants withdrew from the study in the middle. So, we completed the study with ten participants.

The participants have an average of 20 months of software development experience with a minimum of three months and a maximum of 53 months, as shown in \Fref{tab:dev-profile}. Seven of them have experience in developing web applications professionally; that is, outside of learning environments such as classrooms or boot camps, and four of the seven participants experienced in web development have 12 months of experience or more. Further, four of ten participants have experience in developing security features (\eg authentication) for applications. Finally, in terms of programming languages, eight of ten participants use Java frequently, followed by JavaScript and Python, both of which are used by four participants. However, only three participants have experience developing mobile apps professionally, but none for Android. Hence, \textit{participants selected in this study are not expert Android app developers}.

In general, professional developers are likely to learn quickly since many of them have to learn new concepts, tools, and languages on the job.  Also, since commercial software often has to meet security standards, such developers have a better understanding of security-related features than beginners. Further, a non-trivial part of Android app development is done in Java. Hence, despite the lack of exposure to Android app development, \textit{the selected group of participants are representative proxies for Android app developers having experience developing security-related features and can serve as good subjects in this evaluation.}  %Finally, considering this set of developers will be more representative of Android app developers than considering students, which is the norm in usability studies in software engineering ~\cite{Berander:2004, Svahnberg:2008, Iflaah:2015}.

\begin{table}
    \centering
  \ifdef{\TopCaption}{
    \caption{Participant's development experience (in months) in different platforms and with security features.}
  }{}
  \begin{tabular}{rrrrrr} 
    \toprule
    Participant & Software& Web & Mobile & Android & Security\\
    \midrule
    404 & 53 & 17 & 0 & 0 & 17 \\
    704 & 3 & 3 & 0 & 0 & 0 \\
    1004 & 10 & 0 & 8 & 0 & 0 \\
    1503 & 14 & 14 & 0 & 0 & 14 \\
    1803 & 30 & 18 & 6 & 0 & 12 \\
    1904 & 6 & 6 & 0 & 0 & 0 \\
    2103 & 12 & 0 & 0 & 0 & 0 \\
    2403 & 3 & 0 & 2 & 0 & 0 \\
    2703 & 48 & 12 & 0 & 0 & 10 \\
    3103 & 6 & 3 & 0 & 0 & 0 \\
    \bottomrule  
  \end{tabular}
  \ifundef{\TopCaption}{
    \caption{Participant's development experience (in months) in different platforms and with security features.}
  }{}
  \label{tab:dev-profile}
\end{table}

\paragraph{Conducting Interventions:} Since the participants in the study did not know how to make an Android app, we had to train them in Android app development. Hence, we invited all participants to a single 8-hour session. In this session, we introduced them to the fundamental aspects of an Android app (\eg activity, intents). This introduction included a presentation of the necessary concepts, secure coding guidelines, and a live demonstration of how to develop an Android app using Android Studio. After this session, the participants were given Android-related resources such as the Android documentation and free video tutorials (\eg Marakana Android tutorials ~\cite{Marakana:URL}) to learn more about the APIs needed to build Android apps. They were given five days to examine the additional materials independently. After this brief study period, we conducted a group Q\&A session with all the participants. This session comprised two parts. The first part of the session was aimed at helping the participants address any confusion they had about developing an app in Android. The second part of the session was meant to provide them with hands-on experience of Android Studio and Android app development in general.

We designed a similar schedule to train the participants in developing Android apps with SeMA. All the participants were invited to a single 8-hour group session. In this session, we introduced the features in SeMA and demonstrated how to build an app with those features. We provided the participants with the documentation to SeMA and five days to peruse the materials. Finally, we conducted a group Q\&A session with all participants to help them learn about developing an app with SeMA. In addition to the Q\&A, in this session, they made an app with SeMA with our assistance.  The hands-on experience with SeMA helped them get acclimatized to Android app development with SeMA.

This entire process of teaching ten participants Android and SeMA, took approximately two weeks to complete.

In any usability study that requires interventions, the application of the interventions to the users should be designed carefully. In this study, since the training received by a participant could affect the results, \textit{it was necessary to ensure that they received homogeneous instruction}. Hence, we conducted all the training and Q\&A sessions for the participants together instead of individually. However, since each participant was allowed to independently explore Android and SeMA, the effect of individual learning will affect the results of this experiment.

\paragraph{Creating Development Exercises:} We analyzed the selected 13 real-world apps in two ways to design the development exercises for the participants. First, we decompiled the app and manually analyzed its metadata and source code. This analysis revealed the app's screens, any external apps (\eg remote server) the app communicates with, the data flow of the app, and the APIs used in the apps. Second, we installed the app on an Android device and interacted with it to understand its navigation. Based on the analysis of an app's innards and navigation, we constructed a specification of the app. These specifications were used as exercises for the participants in the study. 

The motivation behind using these specifications as exercises for participants was that we wanted to mimic a real-world scenario where a developer would have to start from a description of an app's expected behavior, understand it, and use the relevant features in Android to realize the app. Further, since one of our goals was to measure the ability of SeMA to help developers uncover vulnerabilities, the specifications provided in the exercises were constructed to force participants to make decisions that would affect the app's security. 

As an example, consider the snippet from an exercise assigned to a participant in this experiment -- \textit{"When the app is started, the user is shown a FileUploader screen. The FileUploader screen has a button to allow users to upload a file. Upon clicking this button, the user is shown a list of files. Upon selecting a file, the selected file is saved, and the user is shown a message to indicate if the file was saved successfully."} This specification asks the participant to design an app that allows users to upload a file, a feature that requires the participant to consider security implications. For example, if the participant decided to use an external client to help the user select a file, then she would register the app to receive the chosen file path from the external client via a message. In doing so, an app might open itself to accepting malicious messages. If the app used the file path embedded in a malicious message to create a destination where the selected file will be uploaded, then this decision would enable data injection or directory traversal attack. Hence, the participant's approach to design the file upload feature will impact the app's security. These decisions and SeMA's influence on them was the focus of this study.

% \begin{figure}[ht]
%     \centering
%     \includegraphics[width=\textwidth]{figures/spec-snippet.pdf}
%     \caption{A snippet from an exercise given to developers}
%     \label{fig:spec_snippet}
% \end{figure}

\paragraph{Assigning Exercises to Participants:} We assigned three distinct apps to every participant. Of the three apps assigned to each participant, one baseline app was common to all participants (App ID 1 in \Fref{tab:bench-vul-dist}). All participants developed the baseline app first with SeMA and then without SeMA; that is, using the existing Android development process. We ensured that the baseline app with SeMA and without SeMA was developed on different days to reduce the effect of memorization. The rationale for developing the baseline app with SeMA before developing the same app without SeMA was to avoid the inflation of SeMA's effect due to a participant's prior knowledge about the app/exercise.

In addition to the baseline app, each participant built two apps using SeMA. This assignment strategy resulted in 40 sample implementations of 13 apps. Of the 40 samples, 20 were realizations of the baseline app -- 10 using SeMA and 10 without SeMA. Keeping the baseline app constant across all participants was necessary to measure the effect of SeMA on one app for all participants. The remaining 20 samples were uniformly distributed across the remaining 12 apps (see \Fref{tab:dev-app-map}). This distribution ensured all expected vulnerabilities could be introduced by a participant.

\begin{customTable}
    \centering
  \ifdef{\TopCaption}{
    \caption{Time taken (in minutes) by a participant to make non-baseline apps. The number in brackets indicates if an app was a participant's second or third app.}
  }{}
  \begin{tabular}{@{}rcccccccccccc@{}} 
    \toprule
    & \multicolumn{12}{c}{AppID}\\
    \midrule
     & 2 & 3 & 4 & 5 & 6 & 7 & 8 & 9 & 10 & 11 & 12 & 13\\
    Participant & & & & & & & & & & & &\\
    \midrule
    404 & & & 102 (3)& & & 85 (2)& & & & & &\\
    704 & & & & & & & & 105 (2)& & 135(3)& &\\
    1004 & & & & & & & & & & & 55 (2)&\\
    1503 & 150 (2)& & & & 180 (3)& & & & & & & \\
    1803 & & & & 240  (3)& & & & 90 (2)& & & & \\
    1904 & & 150 (2)& & & & & & & & & &110 (3)\\
    2103 & & & & & 105  (3)& & & & 150 (2)& & & \\
    2403 & & 120 (2)& & & & & & & & & 60 (3)& \\
    2703 & & & & & & & 105 (2)& & & 132 (3)& & \\
    3103 & & & & & & & & & 100 (3)& & & 80 (2)\\
    \bottomrule  
  \end{tabular}
  \ifundef{\TopCaption}{
    \caption{Time taken (in minutes) by a participant to make non-baseline apps. The number in brackets indicates if an app was a participant's second or third app.}
  }{}
  \label{tab:dev-app-map}
\end{customTable}

\paragraph{Observing Development and Collecting Data:}
Each participant worked on the assigned development exercise via remote sessions. Participants decided on the length of each session.  There was no limit on the number of sessions to work on a development exercise.  

In every session, participants shared their screen with us.  Via screen sharing, we observed the development of apps and recorded when participants introduced vulnerabilities that were expected in the developed apps.  Also, we recorded the error messages and causes of error reported by SeMA.

\paragraph{Designing and Administering Development Experience Surveys:} We administered two surveys after the participants finished their exercises.  The purpose of the first survey was to build a participant profile. It asked about the time spent developing software outside formal learning environments (e.g., classrooms), experience developing web and mobile apps, and familiarity with different programming languages and technologies.  The purpose of the second survey was to assess their experience with SeMA.  It asked about various features of SeMA and if they aided/impeded the development workflow. 

\paragraph{Analyzing Observational Data:}
To answer RQ1, we analyzed the development time of developing the baseline app with and without SeMA. The development of baseline apps without SeMA served as the control group, while the development of baseline apps with SeMA served as the treatment group. We computed the average difference in app development time between the control and treatment groups and used a paired two-tailed t-test to measure the effect of SeMA on app development time. 

To answer RQ2, we partitioned the participants into two groups. The +2DX group comprised participants with two or more years of software development experience and the -2DX group comprised participants with less than two years of development experience. We used a two sample two-tailed t-test to compare the average development time using SeMA between these groups to quantify the effect of software development experience on developing software using SeMA.

To answer RQ3, similar to RQ2, we partitioned the participants into two groups. +1SFDX group comprised of participants with one or more years of security feature development experience and -1SFDX group comprised of participants with less than one year of security feature development experience. 

To answer RQ4, we compared the average time taken by participants using SeMA to make apps with expected vulnerabilities stemming from the features of a category and apps without the same features. We used a two sample two-tailed t-test for this comparison.

To answer RQ5, we measured the proportion of total number of expected vulnerabilities introduced by all participants while using SeMA that were then successfully detected by SeMA.  We also considered the proportion of expected vulnerabilities that were prevented by SeMA (due to good defaults). 

While SeMA is effective in detecting expected vulnerabilities, we need to assess the effect of changing the development methodology on the introduction of expected vulnerabilities.  We tackle this concern in RQ6 --- \textit{does the use of SeMA (instead of the usual Android app development process) increase the likelihood of introducing expected vulnerabilities?}

To answer RQ6, we measured the average proportion of expected vulnerabilities introduced by participants in the baseline app with and without SeMA. We compared the average proportions with a paired two-tailed t-test to determine if there was any significant difference in the average proportion of vulnerabilities introduced in the baseline app with and without SeMA. %We used the expected vulnerabilities introduced in the baseline app to answer this RQ since this app was the only app that was made by all participants with and without SeMA.

To answer RQ7, we computed the proportion of expected vulnerabilities introduced while using SeMA for both +2DX and -2DX groups. We compared these two proportions using a two sample two-tailed z-test to measure the effect of the participants' software development experience on introducing vulnerabilities with SeMA.

To answer RQ8, we compared the proportion of expected vulnerabilities introduced while using SeMA by participants in +1SFDX and -1SFDX group. We used a two sample two-tailed z-test for the comparison.

\subsubsection{Results from Observational Data} 

\hfill

\textit{RQ1: Does SeMA affect app development time?} Every participant took less time to make the baseline app with SeMA compared to without SeMA, as shown in columns 2 and 3 of \Fref{tab:dev-time-dist}. The median development time of an app with SeMA is 188 minutes and without SeMA is 283 minutes, which shows that, for the participants considered in this study, \textit{developing the baseline app with SeMA took 40\% less than developing the same app without SeMA}. 

Based on paired two-tailed t-test, SeMA has a significant effect (p-value = 0.0002 with 95\% confidence interval [79 mins., 177 mins.]]) on the mean development time in developing the baseline app. The 95\% confidence interval of SeMA's effect on decreasing the development time suggests that SeMA has the potential to reduce app development time.

% In general, the 95\% confidence interval of the time to develop an app with SeMA (see \Fref{tab:dev-time-dist}) indicates that for 95\% of the samples, \textit{a developer will take a minimum of 136 minutes and a maximum of 159 minutes to make an app of complexity similar to the apps selected in this study.} 

\smallskip \textit{RQ2: Does software development experience affect development time while using SeMA?} The three participants in the +2DX group took 142 minutes on average, to develop an app with SeMA. This time was 150 minutes for the seven participants in the -2DX group. Based on two sample two-tailed t-test comparing the means of these two groups, there is no significant difference in the average time taken by participants from these two groups while using SeMA (p-value = 0.75 with 95\% confidence interval [-55 mins., 39 mins.]). Hence, \textit{software development experience (not related to Android app development) does not affect the time taken by a developer to make an app with SeMA}.  

\smallskip \textit{RQ3: Does security-related feature development experience affect development time while using SeMA?} The four participants in the +1SFDX group took 160 minutes on average, to make an app with SeMA. This time was 139 minutes for the six participants in the -1SFDX group. While participants in the +1SFDX group took more time than participants in the -1SFDX group in this sample, the difference is not significant based on two sample two-tailed t-test (p-value = 0.37 with 95\% confidence interval [-28 mins., 70 mins.]). Hence, \textit{security-related feature development experience does not affect the time taken by a developer to make an app with SeMA.}

\begin{table}
    \centering
    \begin{tabular}{rrrrrrr}
    \toprule
    Participant & \multicolumn{2}{c}{Time to} & \multicolumn{2}{c}{Time to} & Expected & Added\\
     & \multicolumn{2}{c}{Baseline App} & 2nd App & 3rd App & Vulns. & Vulns.\\
     & no SeMA & SeMA & \multicolumn{2}{c}{SeMA} & &\\
    \midrule
    404 & 276 & 165 & 85 & 102 & 15 & 10 \\
    704	& 440 & 255 & 105 & 135 & 14 & 4 \\
    1004 & 180 & 140 & 55 & 135 & 22 & 10 \\
    1503 & 588 & 315 & 180 & 150 & 12 & 5 \\
    1803 & 360 & 210 & 90 & 240 & 16 & 7 \\
    1904 & 290 & 210 & 150 & 110 & 24 & 16 \\
    2103 & 390 & 210 & 150 & 105 & 12 & 5 \\
    2403 & 210 & 138 & 120 & 60 & 21 & 12 \\
    2703 & 240 & 150 & 105 & 132 & 19 & 11 \\
    3103 & 240 & 135 & 80 & 100 & 22 & 11 \\
    \bottomrule
    Average & 321 & 193 & 112 & 138 & & \\
    \bottomrule
    95\% C.I & & \multicolumn{3}{c}{\hspace{3mm}------------[124,170]------------} & & \\
    \bottomrule
    \end{tabular}
    \caption{Observational Data of each participant. The unit of time is minute. Only the baseline app was made with and without SeMA. The non-baseline apps were made only with SeMA. \textit{Expected Vulns.} indicates no. of vulnerabilities that could have been introduced by a participant while developing the three apps assigned to her with SeMA. \textit{Added Vulns.} indicates the no. of vulnerabilities introduced by a participant while developing the three apps assigned to her with SeMA.}
    \label{tab:dev-time-dist}
\end{table}

\smallskip \textit{RQ4: Do the features used in an app affect development time while using SeMA?} With SeMA, for the seven apps with ICC features based expected vulnerabilities (ICC apps), the average development time was 58 minutes more than the average development time of the six apps without ICC-based expected vulnerabilities (non-ICC apps). Based on two sample two-tailed t-test, the 95\% confidence interval of the difference between the average development time to make ICC apps and non-ICC apps ranges from 21 to 95 minutes. A likely reason for longer development time for ICC-apps is that ICC-apps tend to have more screens and navigation between the screens. Considering the significant difference in development time for at least one set of features (\eg ICC), \textit{the features used in an app are likely to affect the time taken to make the app with SeMA}. 

Based on the average development time using SeMA, the development of baseline apps took more time than the development of non-baseline apps (see the \textit{average} row in \Fref{tab:dev-time-dist}). The longer development time for the baseline app is likely due to the methodology's novelty, considering the baseline app was the first app participants developed using SeMA. This observation suggests that development time reduced as participants became familiar with SeMA. However, development time did not necessarily decrease consistently across all apps made with SeMA. For example, participant 1803 took more time to develop the third app than the second app. This increase was likely because the third app had more features than the second app, making the third app more complicated. Specifically, the third app was more complicated since it had features related to user registration and authentication, which were absent in the second app. Hence, \textit{the reduction in development time with SeMA seems to be dependent on the familiarity with SeMA and the features used in the app}. 

\begin{table}
    \centering
    \begin{tabular}{rrrrrrr}
    \toprule
    Participant & N & S & T & W & K & D\\
    \midrule
    404 & 3 & 2 & 5 & 0.6 & 0.4 & 0.2\\
    704 & 2 & 0 & 5 & 0.4 & 0 & 0.4\\
    1004 & 2 & 2 & 5 & 0.4 & 0.4 & 0\\
    1503 & 2 & 3 & 5 & 0.4 & 0.6 & -0.2\\
    1803 & 2 & 2 & 5 & 0.4 & 0.4 & 0\\
    1904 & 2 & 3 & 5 & 0.4 & 0.4 & 0\\
    2103 & 1 & 1 & 5 & 0.2 & 0.2 & 0\\
    2403 & 2 & 0 & 5 & 0.4 & 0 & 0.4\\
    2703 & 2 & 1 & 5 & 0.4 & 0.2 & 0.2\\
    3103 & 2 & 1 & 5 & 0.4 & 0.2 & 0.2\\
    \bottomrule
    Average & & & & 0.4 & 0.3 & 0.1\\
    \bottomrule
    \end{tabular}
    \caption{Proportion of expected vulnerabilities introduced in the baseline app by a participant with and without SeMA. N indicates the no. of expected vulnerabilities introduced by a participant in the baseline app without using SeMA. S indicates the no. of expected vulnerabilities introduced by a participant in the baseline app using SeMA. T indicates the no. of expected vulnerabilities in the baseline app. W indicates the proportion of expected vulnerabilities introduced by a participant in the baseline app without SeMA. K indicates the proportion of expected vulnerabilities introduced by a participant in the baseline app with SeMA. D is the difference between W and K.}
    \label{tab:baseline-vul-prop}
\end{table}

\smallskip \textit{RQ5: Does SeMA detect specific (expected) vulnerabilities introduced in an app's storyboard?}
The participants in this study introduced 91 of 177 instances of expected vulnerabilities in 30 app storyboards while using SeMA (see \textit{Added Vulns.} column in \Fref{tab:dev-time-dist}). For example, participants introduced all instances of expected vulnerabilities in the Web category (see \Fref{tab:vul-freq-dist}). SeMA detected and reported a violation for every one of the 91 instances. Hence, \textit{SeMA is highly likely to detect expected vulnerabilities introduced in a storyboard.}

We anticipated SeMA would prevent 38 instances of expected vulnerabilities -- two instances of I1, 16 instances of I4, two instances of I10, and 18 instances of N8 -- due to good defaults provided by SeMA, e.g., certificate pinning is enabled by default. Since participants overrode the defaults in three instances -- two instances of I4 and one instance of I10, 35 of the 38 instances of expected vulnerabilities were prevented by the good defaults in SeMA. Hence, \textit{the defaults built into SeMA are likely to prevent expected vulnerabilities.}

Overall, since SeMA detected and prevented 126 of 177 (71\%) instances of expected vulnerabilities in the considered sample, we can conclude \textit{SeMA is likely to detect and prevent expected vulnerabilities in an app's storyboard.}

Participants did not introduce 51 instances of expected vulnerabilities -- one in Crypto, five in ICC, 17 in Networking, 27 in Storage, and one in System -- while using SeMA. For example, in the Networking category, 15 of 18 instances of W2 were not introduced since participants decided to use HTTPS over HTTP when communicating with a remote server. In Storage category, when participants introduced S2 (writing to external storage), SeMA flagged the vulnerability. Since this informed the participants about the pitfalls of using external storage, they were more cautious with the further use of external storage.  Consequently, only 5 of the 19 expected instances of S1 (reading from external storage) vulnerability were introduced by the participants. Hence, \textit{prior knowledge of vulnerabilities is likely to affect the introduction of vulnerabilities and, consequently, the lower the observed effectiveness of SeMA.}  On the positive side,  unlike passive interventions like lists of vulnerabilities, \textit{SeMA is likely to serve as an active intervention and improve developer awareness about vulnerabilities}. 

% Interestingly, all instances of expected vulnerabilities in the Web and Permission category were introduced. While the permission category had a low number of instances (2 of 2), every instance of an expected vulnerability in the Web category was introduced (42 of 42) since participants failed to specify a whitelist of safe URLs before loading web content from a URL. 

\smallskip\textit{RQ6: Does the use of SeMA (instead of the usual Android app development process) increase the likelihood of introducing expected vulnerabilities}? The average percentage of introducing an expected vulnerability in the baseline app without SeMA was 40\% across all participants (column W in \Fref{tab:baseline-vul-prop}). With SeMA, the average percentage of introducing an expected vulnerability in the baseline app was 30\% (column K in \Fref{tab:baseline-vul-prop}). The median and mean difference between the percentage of expected vulnerabilities introduced with and without SeMA in the baseline app for each participant is 0\% and 10\%, respectively (see column D of \Fref{tab:baseline-vul-prop}). Further, based on paired two-tailed t-test, there is no significant difference between the mean percentage of introducing an expected vulnerability in the baseline app with and without SeMA (p-value = 0.17 with 95\% confidence interval [-6\%, 26\%]). Consequently, \textit{SeMA is not likely to introduce any more or less vulnerabilities compared to the prevalent app development process.}

% Considering individual vulnerabilities captured in Ghera and selected in this study, 21 of them were introduced at least once while developing an app using SeMA (see column I in \Fref{tab:vul-freq-dist}). Further, 18 of them were introduced more than 50\% of the time they were expected to be introduced, based on columns E and I in \Fref{tab:vul-freq-dist}. This observation suggests that \textit{most vulnerabilities considered in this study are likely to be introduced while using SeMA}. 

\begin{table}
    \centering
    \begin{tabular}{lcrrr}
    \toprule
    Category & Benchmark ID & \# Apps (A) & \# Expected (E) & \# Introduced (I) \\
    \midrule
    Crypto & C3 & 1 & 2 & 1\\
    \midrule
    \multirow[t]{7}{*}{ICC} & I1 & 1 & 2 & 0\\
    & I3 & 2 & 4 & 3\\
    & I4 & 4 & 16 & 2\\
    & I7 & 5 & 9 & 9\\
    & I8 & 1 & 2 & 0\\
    & I9  & 1 & 2 & 0\\
    & I10  & 1 & 2 & 1\\
    \hline
    ICC Ave. & & 2.1 & 5.3 & 2.1\\
    \midrule
    \multirow[t]{4}{*}{Networking} & N6 & 1 & 2 & 1\\
    & N7 & 1 & 2 & 1\\
    & N8 & 6 & 18 & 0\\
    & W2 & 7 & 18 & 3\\
    \hline
    Networking Ave. & & 3.8 & 10 & 1.25\\
    \midrule
    P2 & Permission & 1 & 2 & 2\\
    \midrule
    \multirow[t]{4}{*}{Storage} & S1 & 6 & 19 & 5\\
    & S2 & 7 & 20 & 10\\
    & S3 & 3 & 6 & 5\\
    & S4 & 3 & 6 & 4\\
    \hline
    Storage Ave. & & 4.8 & 12.8 & 6\\
    \midrule
    \multirow[t]{2}{*}{System} & Y5 & 1 & 2	& 2\\
    & Y7 & 1 & 1	& 0\\
    \hline
    System Ave. & & 1 & 1.5 & 1\\
    \midrule
    \multirow[t]{7}{*}{Web} & W1 & 4 & 6 & 6\\
    & W3 & 4 & 6 & 6\\
    & W4 & 4 & 6 & 6\\
    & W6 & 4 & 6 & 6\\
    & W8 & 4 & 6 & 6\\
    & W9 & 4 & 6 & 6\\
    & W10 & 4 & 6 & 6\\
    \hline
    Web Ave. & & 4 & 6 & 6\\
    \bottomrule
    Total Ave. & & 3.1 & 6.8 & 3.5\\
    \bottomrule
    \end{tabular}
    \caption{Frequency distribution of each Ghera vulnerability. A indicates no. of apps with an expected vulnerability. E indicates no. of times a vulnerability was expected to be introduced. I indicates no. of times an expected vulnerability was introduced and prevented.}
    \label{tab:vul-freq-dist}
\end{table}

\smallskip \textit{RQ7: Does software development experience affect the vulnerabilities introduced while using SeMA?} If we consider the participants' software development experience, then the three participants in the +2DX group introduced 53\% of 50 expected vulnerabilities while using SeMA. The seven participants in the -2DX group introduced 50\% of 127 expected vulnerabilities while using SeMA. The difference in the percentage of introducing an expected vulnerability between the two groups of participants is not statistically significant (p-value = 0.44 with 95\% confidence interval [-10\%, 22\%]). Hence, \textit{software development experience does not affect the proportion of expected vulnerabilities introduced by a developer in an app while using SeMA}.

\smallskip \textit{RQ8: Does security-related feature development experience affect the vulnerabilities introduced while using SeMA?} The four participants with one year or more experience in security-related feature development introduced 54\% of 62 expected vulnerabilities while using SeMA. The six participants with less than a year's experience introduced 49\% of 115 expected vulnerabilities while using SeMA. Further, there was no significant difference between the proportion of expected vulnerabilities introduced by the two groups of participants (p-value = 0.72 with 95\% confidence interval [-13\%, 17\%]). A possible reason for the absence of a statistically significant difference is, while the participants in the +1SFDX group have prior knowledge of security features, they do not have experience in security feature development in the context of Android apps. Hence, \textit{security-related feature development experience not specific to Android does not affect the proportion of expected vulnerabilities introduced by a developer in an app while using SeMA}.

\subsubsection{Observations from Developement Experience Survey Data} As mentioned previously, we conducted a survey to collect responses from participants about their experience using SeMA. In the following paragraphs, we present our observations based on a qualitative analysis of the participants' responses. The participant responses are summarized in \Fref{tab:dev-survey}

\paragraph{Did the SeMA extension to Android Studio help/impede in Android app development?} All participants found the extensions to Android Studio helpful. These extensions included the extensions added to a storyboard (described in \Fref{sec:method}) and capabilities related to generating code from a storyboard. Seven of the ten participants said that the extensions helped them \textit{a lot}. The remaining three said that the extensions aided their workflow \textit{quite a bit}. On the flip side, all participants, except one, said that the extensions did not impede their development. Even the one exception said that the extensions only impeded \textit{a little}. Hence, \textit{the participants in this study found the SeMA extension to Android Studio to be largely useful}.

\paragraph{Did the property annotation feature of SeMA help/impede in Android app development?} Eight of the ten participants found that the property annotations in a storyboard (\eg constraints on transitions) helped them in Android app development \textit{a lot}. The other two felt that the annotations helped them \textit{quite a bit}. Only one of the ten participants said that the annotations impeded her development \textit{a little}. \textit{The participant responses show that the property annotation feature largely aided in app development and did not impede their workflow}.

\paragraph{Did the property checking feature of SeMA help/impede in Android app development?} All the participants, except one, said that the property checking feature in SeMA helped them \textit{a lot}. The one exception said that it helped her \textit{quite a bit}. However, five participants said that it impeded their development \textit{a little} and \textit{quite a bit} while the remaining five said that the property checking feature did not impede at all. Since the property checking feature in SeMA forces developers to fix property violations before moving to the next step, a few participants found it impeding their workflow. Hence, \textit{while property checking was helpful for most participants; there is room for improving the process of reporting and fixing violations}.

\paragraph{Did the pre-defined properties provided by SeMA help/impede in Android app development?} Six of the ten participants felt that the pre-defined properties (\eg resources) helped their development \textit{a lot}, and four of them felt that they helped \textit{quite a bit}. In a similar vein, six of them found that the pre-defined properties did not impede their workflow at all, and four of them said that they impeded \textit{a little}. Hence, \textit{participants felt that the pre-defined properties largely aided their workflow}. The four participants who found the pre-defined properties to be \textit{a little} invasive said that it was because they had to keep referring to the documentation for understanding how to use them. This problem can be addressed by adding more auto-completion support for SeMA. 

\paragraph{Feedback: What would you change in SeMA?} All the participants agreed that storyboard-driven development helped them in their development. They particularly appreciated the visualization of the storyboard as it helped them conceptualize the app's behavior. Further, most of them said that SeMA helped them uncover vulnerabilities in the app's design that they would not have otherwise discovered. Typical feedback from participants was to improve the code-completion aid for SeMA and provide access to documentation in the IDE.

\begin{table}
    \centering
    \begin{tabular}{lrrrr}
    \toprule
    & \multicolumn{4}{c}{Scale}\\
    Question & None & A little & Quite a bit & A lot\\
    \midrule
    How much did the SeMA extension to Android Studio & 0 & 0 & 3 & 7\\
    \textit{help} in Android app development? & & & & \\
    How much did the SeMA extension to Android Studio & 9 & 1 & 0 & 0\\
    \textit{impede} in Android app development? & & & & \\
    How much did the property annotation feature of SeMA & 0 & 0 & 2 & 8\\
    \textit{help} in Android app development? & & & &\\
    How much did the property annotation feature of SeMA & 9 & 1 & 0 & 0\\
    \textit{impede} in Android app development? & & & &\\
    How much did the pre-defined properties provided by SeMA & 0 & 0 & 4 & 6\\
    \textit{help} in Android app development? & & & &\\
    How much did the pre-defined properties provided by SeMA & 6 & 4 & 0 & 0\\
    \textit{impede} in Android app development? & & & &\\
    How much did the property checking feature of SeMA & 0 & 0 & 1 & 9\\
    \textit{help} in Android app development? & & & &\\
    How much did the property checking feature of SeMA & 5 & 5 & 0 & 0\\
    \textit{impede} in Android app development? & & & &\\
    \bottomrule
    \end{tabular}
    \caption{Questions asked to participants in a survey. The number in each cell indicate the number of participants who chose a particular scale for a question.}
    \label{tab:dev-survey}
\end{table}

\subsubsection{Threats to Validity}
While the results from this study provide useful insights about the usability of the SeMA methodology, the small number of participants and the small number of apps made by each participant used in this study might affect the generalize-ability of the results. This limitation can be addressed by repeating the experiment with a large number of participants or increasing the number of apps made by a participant. 

Despite the participants in this study having varied experience, the selected participants might not be representative of all kinds of real-world developers. This limitation can be addressed by repeating the study with a more varied sample of developers.

While the participants in this study are reasonable proxies for average Android app developers, they did not have Android app development experience outside of learning environments. Hence, the results and observations from this study might change if repeated with participants with experience in Android app development.

The representation of the Ghera vulnerabilities in the sample of real-world apps selected in this study was not uniform; that is, some vulnerabilities appeared more than others. This lack of uniformity could have affected the results. This concern can be addressed by repeating the experiment with a sample of real-world apps containing a more balanced distribution of vulnerabilities. 

While the exercises assigned to the participants were based on real-world apps, it is possible that they were influenced by our knowledge of vulnerabilities that occur in Android apps. This influence could have introduced bias for certain vulnerabilities in the exercises. This limitation can be addressed by repeating the experiment with a different set of exercises. 

While we ensured that each participant received the same intervention regarding Android app development and SeMA, their personal capacities in grokking new material might have affected the way they made an app. This difference could have influenced the final results.

Finally, we silently observed each participant when they were developing an app. This environment might have caused some participants to behave differently, which might have impacted the way they made an app.  The influence of such factors can be verified by conducting studies that consider environment-related aspects.

\section{Open Challenges}

\paragraph{Performance} The current realization and evaluation of SeMA has focused on ensuring the correctness of generated code. So, while SeMA adheres to performance guidelines outlined in Android’s documentation in the generated code, the generated code may not be performant. This concern can be verified by evaluating the runtime performance of generated code.

\paragraph{Property preservation} One challenge that remains to be addressed in the current realization of SeMA is ensuring the generated implementation satisfies the security properties satisfied by the storyboard. This hinges on ensuring (1) the integrity of generated code and (2) developer-added business logic does not violate the security properties verified in the app storyboard. 
The current realization of SeMA deters developers from modifying the generated code. The generated code is kept separate from the developer-added code in Android Studio. If a developer modifies generated code, a warning message is shown to the user as a pop-up. If the developer modifies the generated code in spite of the warning, then the code is re-generated when the developer compiles the app. While these methods discourage the developer from modifying the generated code, they do not prevent it. One way to prevent modifications to generated code is to use techniques (\eg fingerprinting) to check and enforce the integrity of the generated code.

While SeMA deters the modification of generated code, developers can add business logic code in a way that may not guarantee property preservation. This challenge can be tackled by inhibiting the execution of an app upon detecting the violation of security properties (using techniques such as runtime checks and app sandboxing). However, the current realization of SeMA does not address this concern.

\paragraph{Model modification} Updating code after modifying the model, from which the code was derived, is a challenge in MDD ~\cite{Bran:2012}. Since SeMA is based on MDD, this challenge applies to SeMA as well. In the current realization of SeMA, when the storyboard is updated, code (including developer-added code) is re-generated. While this approach makes it difficult to update the storyboard, it can be addressed by extending SeMA with a mechanism to identify and retain developer-added code associated with parts of the storyboard that haven't been modified. 

\section{Artifacts}
\label{sec:art}
The current realization of SeMA along with the instructions to build and use it is available in the public repository \url{https://bitbucket.org/secure-it-i/sema/src/master/}.

The Ghera benchmarks re-built with SeMA as listed in \Fref{tab:results} are available in the public repository \url{https://bitbucket.org/secure-it-i/sema/src/master/ghera-apps/}.

The raw data collected during the usability study, along with the statistical tests used to analyze the collected data are available in the public repository \url{https://bitbucket.org/secure-it-i/sema/src/master/usability-test/}.  

\section{Related Work}

The current prevalent approach to detect vulnerabilities in an Android app is based on analyzing an app's source code ~\cite{Sufatrio:2015}. The proposed methodology is different from this approach since it enables the detection of known vulnerabilities at the design stage of the app development process, thereby preventing the creation of vulnerabilities.

Prior research has explored solutions based on MDD to simplify mobile app development ~\cite{Vaupel:2016}, reduce technical complexity and development cost ~\cite{Hemel:2011}, and enable the development of cross-platform mobile apps ~\cite{Heitkotter:2015, Brambilla:2014, Fatima:2019}.  However, these solutions are not focused on securing mobile apps. Unlike such existing MDD-based approaches to app development, the primary purpose of SeMA is to enable the reasoning of an app's security properties while designing an app.

Most academic efforts related to model-driven development of mobile apps have not explicitly considered security requirements.  Few commercial efforts related to model-driven development of mobile apps have considered security requirements limited to specific security aspects, \eg access control. For instance, the Mendix App Platform allows developers to specify models for app components, define user roles and accesses w.r.t to the app components, and execute the models in a runtime environment to check for access inconsistencies ~\cite{Martin:2010}. In comparison, SeMA enables the verification of properties about different security aspects albeit limited to known vulnerabilities.

Most prior MDD efforts in software and mobile app development are based on existing software artifacts, \eg UML, finite state machines. Basin \etal ~\cite{Basin:2002} extended UML to SecureUML to enable the formal specification of systems' access control requirements. Usman \etal ~\cite{Usman:2014} relied on use case diagrams, UML class diagrams, and UML state machines in Mobile Apps Generator (MAGS) methodology to capture the specification of an app's requirements, structure, and behavior.  Francese \etal ~\cite{Francese:2015} used finite state machines as a modeling substrate to capture the data-flow, control-flow, and user interactions in an app. In AXIOM methodology, Xiaoping \etal ~\cite{Xiaoping:2013} used the Abstract Model Tree (AMT) representation to model platform-independent app behavior and requirements.
%The IBM Rational Rhapsody enabled developers to specify an app's behavior via class diagrams and provides a runtime environment to execute the class diagrams ~\cite{Rhapsodhy:URL}. \vr{Is this related to mobile app development?} 
Unlike these solutions, instead of relying on artifacts (\eg UML models, finite state machines) that are alien to current mobile app development workflows, SeMA is based on storyboards, an essential artifact in the design and development of mobile apps.

Prior research efforts have highlighted the need for considering security at the early stages of the software development process ~\cite{Geer:2010, IeeeCS:URL, Assal:2018}. This need has manifested in the development of \emph{security requirements engineering}: identifying a system's security-related requirements and checking if the system design satisfies the identified requirements.  In this spirit, Hayley \etal ~\cite{Haley:2008} proposed a framework to define security requirements, context, and assumptions of a system and validate these requirements via formal and informal structured arguments. Haralmbos \etal ~\cite{Haralambos:2007} proposed Secure Tropos, a software development methodology that combines requirements engineering concepts with security-related concepts to aid the design and development of secure software systems.  Riaz \etal ~\cite{Riaz:2016} proposed DIGS framework to help requirements analysts define security goals and verify the completeness and correctness of the requirements w.r.t the defined goals.  SeMA is similar to such solutions as it enables specification and verification of security properties of an app's design.  However, unlike these solutions, the objective of SeMA is to verify security-related requirements stemming from known vulnerabilities.  While this objective limits the security-related requirements that can be tackled with SeMA, it reduces the burden of identifying security-related requirements when using SeMA.  %Besides, SeMA can be easily extended to include security-related requirements based on newly discovered vulnerabilities.  

% For example, Heitkotter \etal ~\cite{Heitkotter:2015} developed MD$^2$, an MDD framework for making cross-platform apps. In this framework, a developer can describe an app in a platform-independent textual DSL and eventually generate platform-specific source code from the specification. Brambilla \etal ~\cite{Brambilla:2014, Fatima:2019} extended the Information Flow Modeling Language (IFML), a standard for depicting UI behavior, to enable the specification of a mobile app's GUI in a platform-independent manner. MobML is a collaborative framework for the design and development of data-intensive applications ~\cite{Hemel:2011}. It offers four modeling languages, each addressing a different aspect of a mobile app (\eg Navigation, UI, Content, and Business Logic), along with a code generator that translates the models into source code for the target platform. Vaupel \etal ~\cite{Vaupel:2016} propose a model-driven approach for developing mobile business apps that support the configuration of user role variants. In this approach, developers can configure variants of an app, generated based on the app's various user roles. 

\section{Conclusion}
In this paper, we explored an alternative (and complementary) approach to develop secure Android apps. This approach focuses on preventing vulnerabilities instead of the traditional approach of detecting vulnerabilities after they have occurred. To this end, SeMA is a design-based methodology based on Model-Driven development and existing design techniques to help build secure Android apps. 

SeMA extends storyboards with features that enable developers and designers to collaborate and specify an app's behavior iteratively, while also reasoning about and verifying security properties related to confidentiality and integrity in an app's design. Furthermore, SeMA has code generation support that helps translate annotated/extended storyboards, specified in SeMA, to an implementation. Developers can enrich the generated implementation with business logic code in Java or Kotlin. 

A proof-of-concept realization of SeMA is available for Android Studio. An empirical evaluation of SeMA shows that SeMA can prevent 49 of 60 vulnerabilities captured in the Ghera benchmark suite through a combination of information flow analysis, rule checking, and code generation techniques.

A usability study of SeMA with ten professional software developers shows that SeMA is highly likely to reduce the cost of development and prevent known Android app vulnerabilities. The same study showed that the vulnerabilities prevented by SeMA are highly likely to be introduced by developers. Hence, SeMA is not only useful in preventing vulnerabilities, but it is also relevant.  
\section{Acknowledgement}

We would like to acknowledge the support from the Android Security Rewards program, which partially funded the research conducted for this paper. The Android Security Rewards program plays a pivotal role in advancing our understanding of mobile security, and we appreciate the opportunity to contribute to this important field.

\bibliographystyle{acm}
\bibliography{refs}

\appendix
\section{Formal Syntax of SeMA}
\label{sec:syntax}

Before specifying the syntax of the storyboard language, we first list the sets needed to specify the syntax. Boolean Values BV = \{\textit{true}, \textit{false}\}, non-boolean values V, Apps APP, Screens S, Proxy screens PS, App Identifiers APPID, Screen Identifiers SID, URIs U, URI strings US, Widgets W, Widget Identifiers WID, Widget Types WT, Gestures on widgets G, User actions UA, Transitions TR, Transition Identifiers TID, Boolean expressions B, Resources R, Resource Names RN, Capabilities provided by resources C, Capability Names CN, argument to operations ARG, parameter names of screens VID, Operations F, Operation Names FN, and Input Parameters of screens P.

Assume that the syntactic structure of values, identifiers, names, and URI strings is given. For example, identifiers consist of non-empty string of letters. The remaining syntactic sets are defined inductively via formation rules shown below. 
\begin{figure}
    \centering
    \begin{grammar}
<\textit{app}> ::= \textit{appId} s r?

<\textit{s}> ::= \keyword{screen} \textit{sid} u? w tr? | \keyword{proxy} \keyword{safe}? \textit{sid} \textit{appId}? \textit{us}

<\textit{w}> ::= \keyword{safe}? wt wid (\textit{v} | f | \textit{vid}) | $w_0w_1$

<\textit{wt}> ::= \keyword{TextView} | \keyword{EditText} | \keyword{Button} | \keyword{WebView}

<\textit{f}> ::= \keyword{fun} \textit{fn} (\textit{rn.cn})? a?

<\textit{r}> ::= (\keyword{access} \keyword{all} | \keyword{user} | \keyword{own}) resource rn c | $r_0r_1$

<\textit{c}> ::= \keyword{priv}? \textit{cn} f | $c_0c_1$

<\textit{a}> ::= (\keyword{safe}? \textit{v} | \textit{vid} | \textit{wid} | f) | $a_0a_1$

<\textit{tr}> ::= \keyword{transition} \textit{tid} \keyword{dest} \textit{sid} \keyword{cond} ua (\keyword{and} b)? p? | $tr_0tr_1$

<\textit{ua}> ::= \textit{wid}.g

<\textit{g}> ::= \keyword{click} | \keyword{swipe} | \keyword{drag}

<\textit{p}> ::= \keyword{param} \textit{vid} (\textit{wid} | $vid_0$ | \textit{v} | f) | $p_0p_1$

<\textit{u}> ::= \textit{us} | \textit{us/k} | $u_0u_1$

<\textit{k}> ::= $vid$ | $k_0k_1$

<\textit{b}> ::= \textit{bv} | f | $b_0$ \keyword{and} $b_1$ | $b_0$ \keyword{or} $b_1$ | \keyword{not} b
\end{grammar}
    \caption{Formation Rules}
    \label{fig:gram}
\end{figure}

The formation rules use the following meta-variables to range over the syntactic sets:

\textit{bv} ranges over BV, \textit{v} ranges over V, \textit{app} ranges over APP, \textit{appid} ranges over APPID, \textit{vid} ranges over VID, \textit{wid} ranges over WID, \textit{sid} ranges over SID, \textit{a} ranges over ARG, \textit{s} ranges over S, \textit{ps} ranges over PS, \textit{w} ranges over W, \textit{wt} ranges over WT, \textit{tr} ranges over TR, \textit{b} ranges over B, \textit{r} ranges over R, \textit{rn} ranges over RN, \textit{f} ranges over F, \textit{p} ranges over P, \textit{fn} ranges over FN, \textit{tid} ranges over TID, \textit{u} ranges over U, \textit{us} ranges over US, \textit{c} ranges over C, \textit{cn} ranges over CN, and \textit{g} ranges over G.

The meta-variables can be sub-scripted or primed e.g. $s^\prime$, $s_0$ stands for an element in the set S. The symbol $?$ is used to denote an \textit{optional} syntactic construct. The formation rules are presented in a variant of BNF (Backus-Naur form) in \Fref{fig:gram}.

% \begin{lstlisting}[escapeinside={(*}{*)},morekeywords={screen, TextView, EditText, Button, WebView, fun, resource, safe, priv, transition, dest, cond, and, or, not, param, proxy, click, swipe, drag, access, all, user, own},keywordstyle=\color{blue},showstringspaces=false]
%   app ::= appId (s) (r)(*$^?$*)
%   s ::= screen sid (u)(*$^?$*) w (tr)(*$^?$*) (*$|$*) proxy (safe)(*$^?$*) sid (appId)(*$^?$*) us (*$|$*) (*$s_0s_1$*)
%   w ::= (safe)(*$^?$*) wt wid (v  (*$|$*) f (*$|$*) vid) (*$|$*) (*$w_0w_1$*)
%   wt ::= TextView (*$|$*) EditText (*$|$*) Button (*$|$*) WebView
%   f ::= fun fn (rn.cn)(*$^?$*) (a)(*$^?$*)
%   r ::= access av resource rn c (*$|$*) (*$r_0r_1$*)
%   av :: all (*$|$*) user (*$|$*) own
%   c ::= (priv)(*$^?$*) cn f (*$|$*) (*$c_0c_1$*)
%   a ::= (safe)(*$^?$*) (v (*$|$*) vid (*$|$*) wid (*$|$*) f) (*$|$*) (*$a_0a_1$*)
%   tr ::= transition tid dest sid (cond ua and b) (p)(*$^?$*) (*$|$*) (*$tr_0tr_1$*)
%   ua ::= wid.g
%   g ::= click (*$|$*) swipe (*$|$*) drag
%   p ::= param vid (wid (*$|$*) v (*$|$*) f) (*$|$*) param (*$vid_0$*) (*$|$*) (*$p_0p_1$*)
%   u ::= us (*$|$*) us/{vid} (*$|$*) (*$u_0u_1$*)
%   b ::= (*\textit{true}*) (*$|$*) (*\textit{false}*) (*$|$*) f (*$|$*) (*$b_0$*) and (*$b_1$*) (*$|$*) (*$b_0$*) or (*$b_1$*) (*$|$*) not b

% \end{lstlisting}

\section{Formal Semantics of SeMA}
\label{sec:semantics}
Before describing the small-step operational semantics of a storyboard in SeMA, we introduce the function and operators in the meta-language that will be needed to understand the semantics.

\begin{enumerate}

    \item $\phi : ID \rightarrow SID$, where $ID = VID \cup WID$ is a function that is used to keep track of the screen associated with a given widget or variable identifier. For a widget identifier $wid \in WID$, it returns a screen identifier $sid$ where $wid$ is a widget in $sid$. For a variable identifier $vid \in VID$, it returns a screen identifier $sid$ such that $vid$ as an input parameter of $sid$.
    
    \item $gen(id,sid)$ is an operator that takes a pair $(id,sid) : id \in ID, sid \in SID$ and returns a new $id\prime \in ID$.

    \item $\sigma : ID \rightharpoonup V$ is a partial function that maps IDs to Values and is used to denote the state of the app. This function keeps track of the values assigned to widget and variable identifiers in a screen.
    
    \item $init : APPID \rightarrow SID$ is a function that maps an \textit{appId} to the start screen identifier \textit{sid} of an app. The start screen of an app is the screen that the user sees when the app is started for the first time. This function returns the screen ID in an app that needs to be displayed when the app is started. 
    
    \item $eval : F \rightarrow X$, where $X = V \cup BV$ is a function that returns a \textit{non-boolean value} or a \textit{boolean value} for an operation $f \in F$. This function interprets an operation used in the storyboard.
    
    \item $evalUriVar : ID \rightharpoonup V$ is a partial function that returns the value associated with a variable identifier associated with the URI of a screen. The value is provided by an external app that uses the URI to trigger the corresponding screen. If no external app provides a value for a variable identifier $id$, then $evalUriVar(id) = \perp$.
    
    \item $outTr : SID \rightarrow TR$ is a function that returns all outgoing transitions $tr \in TR$ from a screen $sid \in SID$. 
    
    \item $occur : UA \rightarrow BV$ is a predicate that is \textit{true} if the user performed a gesture (\eg button click) and is \textit{false} if the gesture was not performed. This predicate is used to capture any gestures a user may perform on the widgets in a screen visible to the user.
    
    \item $validScr : (SID) \rightarrow BV$ is a predicate that is \textit{true} if a given $sid \in SID$ is the ID of a screen in the app; and \textit{false} otherwise.
    
    \item $compr : (h,h^\prime) \rightarrow BV$ is a predicate that is \textit{true} if $\forall x \in dom(h) : (h(x) = v \land h^\prime(x) = \perp) \lor (h(x) = \perp \land h^\prime(x) = v)$; \textit{false} otherwise.
    
    \item $stop : (APPID) \rightarrow BV$ is a predicate that is \textit{true} if an app with $appId \in APPID$ is moved to the background or is killed explicitly by the user or Android; and \textit{false} otherwise.
    
    \item $\rho : (RN,CN) \rightarrow BV$ is a predicate that is \textit{true} if a capability identifier $cn \in CN$ offered by a resource identifier $rn \in RN$ is defined; and \textit{false} otherwise.
    
    \item $erase(h,h^\prime)$ is a binary operator that takes two partial functions h,$h^\prime$ and returns another partial function g such that $\forall x \in dom(h^\prime): g(x) = h^\prime(x)$, but $g(x) = \perp$ when $h(x) \neq \perp \land h^\prime(x) \neq \perp$. The operator is required to modify the state of the app, $\sigma$, when the transition to a screen is taken.
    
    \item $erase_K(h,h\prime)$ is a binary operator that takes two partial functions h,$h^\prime$ and returns another partial function g such that $\forall x \in K, K \subseteq dom(h^\prime) : g(x) = h^\prime(x)$, but $g(x) = \perp$ when $h(x) \neq \perp \land h^\prime(x) \neq \perp$. The operator is used to change the state of the app, $\sigma$, when a transition from a screen to itself is taken.

\end{enumerate}

In addition to the meta-functions and meta-operators, we assume that the custom resources defined in a storyboard are parsed beforehand and stored as the app's custom resources. Hence, resource definitions are not explicitly described in the semantics presented below.

Every syntactic construct defined in \Fref{sec:syntax} is evaluated in a configuration of the form $\st{sid}{\sigma}$, where $sid$ denotes the \textit{current screen} visible to the user of the app and $\sigma$ denotes the state of the app when the user is at screen $sid$. The \textit{initial configuration} of the app $\st{\perp}{\sigma}$.

\begin{table}
    \centering
    \begin{tabular}{c}
    \toprule
     \vbox{\begin{equation}
    \label{eq:app}
    \inferrule
        {init(appId) = sid}
        {\srule{\app{s}{}{r}}{\st{\perp}{\sigma}}{\langle \app{s}{}{r}, \st{sid}{\sigma}\rangle}}
    \end{equation}} \\
    
    \vbox{\begin{equation}
    \label{eq:appStop}
    \inferrule
        {stop(appId) = \textit{true}}
        {\srule{\app{s}{}{r}}{\st{sid}{\sigma}}{\st{\perp}{erase(\sigma,\sigma)}}}
\end{equation}} \\

\vbox{\begin{equation}
    \label{eq:appSid}
    \inferrule
        {stop(appId) = \textit{false} \\ \srule{s}{\st{sid}{\sigma}}{\st{sid^\prime}{\sigma^\prime}}}
        {\srule{\app{s}{}{r}}{\st{sid}{\sigma}}{\langle \app{s}{}{r}, \st{sid^\prime}{\sigma^\prime}\rangle}}
\end{equation}} \\
    \bottomrule
    \end{tabular}
    \caption{App-related Semantic Rules.}
    \label{tab:app-related-rules}
\end{table}

\paragraph{App-related rules} The rules listed in \Fref{tab:app-related-rules} describe an app's behavior when a user starts an app, interacts with the app, or closes the app. 

When an app is started, as per rule \ref{eq:app}, it moves from its initial configuration to a configuration where the current screen is set to $sid \in SID$. $sid$ is obtained from the meta-function $init$.

As per rule \ref{eq:appStop}, if an app is stopped by the user or Android, then the terminal configuration $\st{\perp}{\sigma}$, such that $\forall x \in dom(\sigma) : \sigma(x) = \perp$, is reached. This configuration indicates that no screen of the app is visible to the user and the state of the app is undefined.

As per rule \ref{eq:appSid}, an app in a configuration $\st{sid}{\sigma}$ moves to a new configuration if a screen $sid$ in the app changes the configuration.

\begin{table}
    \centering
    \begin{tabular}{c}
    \toprule
     \vbox{\begin{equation}
    \label{eq:screenSeq0}
    \inferrule
        {\srule{s_0}{\st{sid}{\sigma}}{\st{sid^\prime}{\sigma^\prime}}}
        {\srule{s_0s_1}{\st{sid}{\sigma}}{\st{sid^\prime}{\sigma^\prime}}}
\end{equation}} \\
    
    \vbox{\begin{equation}
    \label{eq:screenSeq1}
    \inferrule
        {\srule{s_1}{\st{sid}{\sigma}}{\st{sid^\prime}{\sigma^\prime}}}
        {\srule{s_0s_1}{\st{sid}{\sigma}}{\st{sid^\prime}{\sigma^\prime}}}
\end{equation}} \\

\vbox{\begin{equation}
    \label{eq:screen}
    \inferrule
        {\srule{w}{\st{sid}{\sigma}}{\st{sid}{\sigma^\prime}}\\
        \srule{tr}{\st{sid}{\sigma^\prime}}{\st{sid^\prime}{\sigma^{\prime\prime}}}}
        {\srule{\screen{sid}{u}{w}{tr}}{\st{sid}{\sigma}}{\st{sid^\prime}{\sigma^{\prime\prime}}}}
\end{equation}} \\

\vbox{\begin{equation}
    \label{eq:screenNoTran}
    \inferrule
        {\srule{w}{\st{sid}{\sigma}}{\st{sid}{\sigma^\prime}}}
        {\srule{\screen{sid}{u}{w}}{\st{sid}{\sigma}}{\st{sid}{\sigma^\prime}}}
\end{equation}} \\

\vbox{\begin{equation}
    \label{eq:ps}
    \inferrule
        {}
        {\srule{\ps{sid}{appId}{u}}{\st{sid}{\sigma}}{\st{sid}{\sigma}}}
\end{equation}} \\
    \bottomrule
    \end{tabular}
    \caption{Screen-related Semantic Rules.}
    \label{tab:screen-related-rules}
\end{table}

\paragraph{Screen-related rules.} The rules listed in \Fref{tab:screen-related-rules} apply to an app's screens. They are informally described in \Fref{sec:informal-sem} as screen extensions.

Rule \ref{eq:screenSeq0} is applicable when the app is in a configuration $(sid,\sigma)$ and $sid$ is the identifier of the first screen $s_0$ in the sequence $s_0s_1$.

Rule \ref{eq:screenSeq1} is applicable when the app is in a configuration $(sid,\sigma)$ and $sid$ is the identifier of some screen in the sequence $s_0s_1$.

As per rule \ref{eq:screen}, an app moves from current screen $sid$ with state $\sigma$ to a screen $sid\prime$ with state $\sigma\prime\prime$ when the widgets in the screen $sid$ extend the state $\sigma$ to $\sigma\prime$ and one of the outgoing transitions sets the current screen to $sid\prime$ with new state $\sigma\prime\prime$.

As per rule \ref{eq:screenNoTran}, a screen with no outgoing transitions causes the current state $\sigma$ to change to a new state $\sigma\prime$ when the widgets in the screen $sid$ extend the state $\sigma$ to $\sigma\prime$.

A proxy screen does not result in any configuration change as per rule \ref{eq:ps}.

\begin{table}
    \centering
    \begin{tabular}{c}
    \toprule
     \vbox{\begin{equation}
    \label{eq:wdgSeq}
    \inferrule
        {\srule{w_0}{\st{sid}{\sigma}}{\st{sid}{\sigma^\prime}} \\
        \srule{w_1}{\st{sid}{\sigma^\prime}}{\st{sid}{\sigma^{\prime\prime}}}}
        {\srule{w_0w_1}{\st{sid}{\sigma}}{\st{sid}{\sigma^{\prime\prime}}}}
\end{equation}} \\

\vbox{\begin{equation}
    \label{eq:wdgVar}
    \inferrule
        {\srule{x}{\st{sid}{\sigma}}{v}}
        {\srule{\widget{wt}{wid}{x}}{\st{sid}{\sigma}}{\st{sid}{\sigma[gen(wid,sid)\mapsto v]}}}, x ::= vid \hspace{1mm} or \hspace{1mm} f
\end{equation}}\\

\vbox{\begin{equation}
    \label{eq:wdgInvVar}
    \inferrule
        {\srule{vid}{\st{sid}{\sigma}}{\perp}}
        {\srule{\widget{wt}{wid}{vid}}{\st{sid}{\sigma}}{\st{sid}{\sigma}}}
\end{equation}} \\

\vbox{\begin{equation}
    \label{eq:wdgVal}
    \inferrule
        {}
        {\srule{\widget{wt}{wid}{v}}{\st{sid}{\sigma}}{\st{sid}{\sigma[gen(wid,sid)\mapsto v]}}}
\end{equation}} \\
    \bottomrule
    \end{tabular}
    \caption{Widget-related Semantic Rules.}
    \label{tab:widget-related-rules}
\end{table}

\paragraph{Widget-related rules} The rules listed in \Fref{tab:widget-related-rules} correspond to linking widgets to data sources in a screen. These rules are informally described as \textit{Widget Extensions} in \Fref{sec:informal-sem}.

As per rule \ref{eq:wdgSeq}, a sequence of widgets in a screen $sid$ extends the corresponding state $\sigma$ to $\sigma\prime\prime$ if each widget in the sequence extends $\sigma$.

As per rule \ref{eq:wdgVar}, a widget $wid$ extends state $\sigma$ with value $v$ if $wid$ is initialized with a variable $vid$ or an operation $f$, and $vid$ or $f$ evaluates to $v$ in state $\sigma$. 

As per rule \ref{eq:wdgInvVar}, a widget $wid$ does not change the current configuration if $wid$ is initialized with an invalid variable \ie the variable is undefined in $\sigma$.

As per rule \ref{eq:wdgVal}, a widget $wid$ extends state $\sigma$ with value $v$ if $wid$ is initialized with the value $v$.

\begin{table}
    \centering
    \begin{tabular}{c}
    \toprule
     \vbox{\begin{equation}
    \label{eq:trSeq1}
    \inferrule
        {\srule{tr_0}{\st{sid}{\sigma}}{\st{sid^\prime}{\sigma^\prime}} \\ compr(\sigma,\sigma^\prime)}
        {\srule{tr_0tr_1}{\st{sid}{\sigma}}{\st{sid^\prime}{\sigma^\prime}}}
\end{equation}} \\

\vbox{\begin{equation}
    \label{eq:trSeq2}
    \inferrule
        {\srule{tr_0}{\st{sid}{\sigma}}{\st{sid}{\sigma}} \\
        \srule{tr_1}{\st{sid}{\sigma}}{\st{sid^\prime}{\sigma^\prime}}}
        {\srule{tr_0tr_1}{\st{sid}{\sigma}}{\st{sid^\prime}{\sigma^\prime}}}
\end{equation}} \\

\vbox{\begin{equation}
    \label{eq:tr}
    \inferrule
        {\srule{ua}{\st{sid}{\sigma}}{\textit{true}} \\ \srule{b}{\st{sid}{\sigma}}{\textit{true}} \\
        validScr(sid^\prime) = \textit{true} \\
        \srule{p}{\st{sid}{\sigma}}{\st{sid}{\sigma^\prime}} \\ 
        sid \neq sid^\prime}
        {\srule{\trans{tid}{sid^\prime}{ua}{b}{p}}{\st{sid}{\sigma}}{\st{sid^\prime}{erase(\sigma,\sigma^\prime)}}}
\end{equation}} \\

\vbox{\begin{equation}
    \label{eq:trVanilla}
    \inferrule
        {\srule{ua}{\st{sid}{\sigma}}{\textit{true}} \\ \srule{b}{\st{sid}{\sigma}}{\textit{true}} \\
        validScr(sid^\prime) = \textit{true}}
        {\srule{\trans{tid}{sid^\prime}{ua}{b}{}}{\st{sid}{\sigma}}{\st{sid^\prime}{erase(\sigma,\sigma^\prime)}}}
\end{equation}} \\

\vbox{\begin{equation}
    \label{eq:trSelf}
    \inferrule
        {\srule{ua}{\st{sid}{\sigma}}{\textit{true}} \\ \srule{b}{\st{sid}{\sigma}}{\textit{true}} \\
        validScr(sid^\prime) = \textit{true} \\
        \srule{p}{\st{sid}{\sigma}}{\st{sid}{\sigma^\prime}}}
        {\srule{\trans{tid}{sid}{ua}{b}{p}}{\st{sid}{\sigma}}{\st{sid}{erase_{WID}(\sigma,\sigma^\prime)}}}
\end{equation}} \\

\vbox{\begin{equation}
    \label{eq:trFail1}
    \inferrule
        {\srule{ua}{\st{sid}{\sigma}}{\textit{false}}}
        {\srule{\trans{tid}{sid\prime}{ua}{b}{p}}{\st{sid}{\sigma}}{\st{sid}{\sigma}}}
\end{equation}} \\

\vbox{\begin{equation}
    \label{eq:trFail2}
    \inferrule
        {\srule{b}{\st{sid}{\sigma}}{\textit{false}}}
        {\srule{\trans{tid}{sid^\prime}{ua}{b}{p}}{\st{sid}{\sigma}}{\st{sid}{\sigma}}}
\end{equation}} \\
    \bottomrule
    \end{tabular}
    \caption{Transition-related Semantic Rules.}
    \label{tab:trans-related-rules}
\end{table}

\paragraph{Transition-related rules} The rules listed in \Fref{tab:trans-related-rules} specify an app's behavior in the context of outgoing transitions from a screen. They are informally described as \textit{Extensions to Transition} in \Fref{sec:informal-sem}.

As per rule \ref{eq:trSeq1}, if the first transition in an ordered set of transitions is taken, then the remaining transitions are not evaluated.

As per rule \ref{eq:trSeq2}, if the first transition in an ordered set of transitions is not taken, then the remaining transitions in the ordered set are evaluated.

As per rule \ref{eq:tr}, a transition from a screen to a different screen is taken if the user action associated with it evaluates to \textit{true}, the associated boolean condition evaluates to \textit{true}, and the arguments to the destination screen's parameters $p$ extend the current state $\sigma$. Rule \ref{eq:trVanilla} is a variant of this rule for a transition without associated input parameters. Rule \ref{eq:trSelf} is a variant of this rule for a transition from a screen to itself.

As per rule \ref{eq:trFail1}, a transition is not taken if the user action associated with it evaluates to \textit{false}.

As per rule \ref{eq:trFail2}, a transition is not taken if the boolean condition associated with it evaluates to \textit{false}.

\begin{table}
    \centering
    \begin{tabular}{c}
    \toprule
     \vbox{\begin{equation}
    \label{eq:funCall}
    \inferrule
        { \rho(rn,cn) = \textit{true} \\ eval(f) = v}
        {\srule{\fn{fn}{rn.cn}{}}{\st{sid}{\sigma}}{v}}
\end{equation}} \\

\vbox{\begin{equation}
    \label{eq:funCallArg}
    \inferrule
        {\rho(rn,cn) = \textit{true} \\ eval(f) = v}
        {\srule{\fn{fn}{rn.cn}{a}}{\st{sid}{\sigma}}{v}}
\end{equation}} \\

\vbox{\begin{equation}
    \label{eq:funCallArg1}
    \inferrule
        {eval(f) = v}
        {\srule{\fn{fn}{}{a}}{\st{sid}{\sigma}}{v}}
\end{equation}} \\

\vbox{\begin{equation}
    \label{eq:funCallVanilla}
    \inferrule
        {eval(f) = v}
        {\srule{\fn{fn}{}{}}{\st{sid}{\sigma}}{v}}
\end{equation}} \\

\vbox{\begin{equation}
    \label{eq:boolCall}
    \inferrule
        {\rho(rn,cn) = \textit{true} \\ eval(f) = bv}
        {\srule{\fn{fn}{rn.cn}{}}{\st{sid}{\sigma}}{bv}}
\end{equation}} \\

\vbox{\begin{equation}
    \label{eq:boolCallArgs}
    \inferrule
        {\rho(rn,cn) = \textit{true} \\ eval(f) = bv}
        {\srule{\fn{fn}{rn.cn}{a}}{\st{sid}{\sigma}}{bv}}
\end{equation}} \\

\vbox{\begin{equation}
    \label{eq:boolCallArgs1}
    \inferrule
        {eval(f) = bv}
        {\srule{\fn{fn}{}{a}}{\st{sid}{\sigma}}{bv}}
\end{equation}} \\

\vbox{\begin{equation}
    \label{eq:boolCallArgs2}
    \inferrule
        {eval(f) = bv}
        {\srule{\fn{fn}{}{}}{\st{sid}{\sigma}}{bv}}
\end{equation}} \\

    \bottomrule
    \end{tabular}
    \caption{Operation-related Semantic Rules.}
    \label{tab:op-related-rules}
\end{table}

\paragraph{Operation-related rules} The rules listed in \Fref{tab:op-related-rules} are used to evaluate an operation used in the storyboard. These rules correspond to the informal description of \textit{Operations} in \Fref{sec:informal-sem}.

As per rule \ref{eq:funCall}, an operation $f$ evaluates to a value $v$, if the meta-function $eval$ evaluates $f$ to a non-boolean value $v$ and the resource used in the operation is defined. Rule \ref{eq:funCallArg} is a variant of this rule for operations with arguments.

As per rule \ref{eq:funCallArg1}, an operation $f$ evaluates to a value $v$, if the meta-function $eval$ evaluates $f$ to a non-boolean value $v$. Rule \ref{eq:funCallVanilla} is a variant of this rule for operations without resources or arguments.

As per rule \ref{eq:boolCall}, an operation $f$ evaluates to a boolean value $bv$, if the meta-function $eval$ evaluates $f$ to a boolean $bv$ and the resource used in the operation is defined. Rule \ref{eq:boolCallArgs} is a variant of this rule for operation with arguments. 

As per rule \ref{eq:boolCallArgs1}, an operation $f$ evaluates to a boolean value $bv$, if the meta-function $eval$ evaluates $f$ to a boolean $bv$. Rule \ref{eq:boolCallArgs2} is a variant of this rule for operations without resources or arguments.

\begin{table}
    \centering
    \begin{tabular}{c}
    \toprule

    \vbox{\begin{equation}
    \label{eq:widgetRef}
    \inferrule
        {\sigma(gen(wid,sid)) = v}
        {\srule{wid}{\st{sid}{\sigma}}{v}}
\end{equation}} \\

    \vbox{\begin{equation}
    \label{eq:ua}
    \inferrule
        {\srule{occurs(wid,g)}{\st{sid}{\sigma}}{bv}}
        {\srule{wid.g}{\st{sid}{\sigma}}{bv}}
\end{equation}} \\
    
    \vbox{\begin{equation}
    \label{eq:varRef}
    \inferrule
        {\sigma(gen(vid,sid)) = v}
        {\srule{vid}{\st{sid}{\sigma}}{v}}
\end{equation}} \\

\vbox{\begin{equation}
    \label{eq:varRef1}
    \inferrule
        {\sigma(gen(vid,sid)) = \perp \\ evalUriVar(gen(vid,sid)) = v}
        {\srule{vid}{\st{sid}{\sigma}}{v}}
\end{equation}} \\
    
    \vbox{\begin{equation}
    \label{eq:varInvRef}
    \inferrule
        {\sigma(gen(vid,sid)) = \perp \\ evalUriVar(gen(vid,sid)) = \perp}
        {\srule{vid}{\st{sid}{\sigma}}{\perp}}
\end{equation}} \\
    
    \vbox{\begin{equation}
    \label{eq:widgetInvRef}
    \inferrule
        {\sigma(gen(wid,sid)) = \perp}
        {\srule{wid}{\st{sid}{\sigma}}{\perp}}
\end{equation}} \\
    
    \bottomrule
    \end{tabular}
    \caption{ID-related Semantic Rules.}
    \label{tab:id-related-rules}
\end{table}

\paragraph{ID-related rules} The rules listed in \Fref{tab:id-related-rules} are used to evaluate a widget or variable identifiers. 

As per rule \ref{eq:widgetRef}, a widget identifier $wid$ evaluates to a value $v$ under state $\sigma$ if the state $\sigma$ maps $wid$ to the value $v$.

As per rule \ref{eq:ua}, a gesture \textit{g} on a widget $wid$ evaluates to \textit{true}, if the gesture occurs on the widget; \textit{false} otherwise.

As per rule \ref{eq:varRef}, a variable identifier $vid$ evaluates to a value $v$ under state $\sigma$ if the state $\sigma$ maps $vid$ to the value $v$.

As per rule \ref{eq:varRef1}, a variable identifier $vid$ evaluates to a value $v$, if $\sigma(vid)$ is undefined, but an external app has provided the value $v$ for $vid$ via $evalUriVar$.

As per rule \ref{eq:varInvRef}, a variable evaluates to $\perp$ if the variable identifier is undefined in $\sigma$ and no external app has provided a value for the variable.

As per rule \ref{eq:widgetInvRef}, a widget evaluates to $\perp$ if the widget identifier is undefined in $\sigma$. 

\begin{table}
    \centering
    \begin{tabular}{c}
    \toprule

    \vbox{\begin{equation}
    \label{eq:pSeq}
    \inferrule
        {\srule{p_0}{\st{sid}{\sigma}}{\st{sid}{\sigma^\prime}} \\
        \srule{p_1}{\st{sid}{\sigma^\prime}}{\st{sid}{\sigma^{\prime\prime}}}}
        {\srule{p_0p_1}{\st{sid}{\sigma}}{\st{sid}{\sigma^{\prime\prime}}}}
\end{equation}} \\

    \vbox{\begin{equation}
    \label{eq:pVal}
    \inferrule
        {}
        {\srule{\param{vid}{v}}{\st{sid}{\sigma}}{\st{sid}{\sigma[gen(vid,\phi(vid))\mapsto v]}}}
\end{equation}} \\
    
    \vbox{\begin{equation}
    \label{eq:pVar}
    \inferrule
        {\srule{x}{\st{sid}{\sigma}}{v}}
        {\srule{\param{vid_0}{x}}{\st{sid}{\sigma}}{\st{sid}{\sigma[gen(vid_0,\phi(vid_0))\mapsto v]}}}, x ::= vid_1, wid, \hspace{1mm} or \hspace{1mm} f
\end{equation}} \\
    
    \vbox{\begin{equation}
    \label{eq:pInvVar}
    \inferrule
        {\srule{x}{\st{sid}{\sigma}}{\perp}}
        {\srule{\param{vid_0}{x}}{\st{sid}{\sigma}}{\st{sid}{\sigma}}}, x ::= vid_1 \hspace{1mm} or \hspace{1mm} wid
\end{equation}} \\
    
    \bottomrule
    \end{tabular}
    \caption{Screen's input parameter-related Semantic Rules.}
    \label{tab:screen-inp-related-rules}
\end{table}

\paragraph{Screen's input parameter-related rules} The rules listed \Fref{tab:screen-inp-related-rules} are used to evaluate the arguments provided to input parameters of a screen as part of transitions to that screen. An informal description of these rules is provided in \Fref{sec:informal-sem} as \textit{Extensions to Transitions}.

As per rule \ref{eq:pSeq}, a sequence of syntactic constructs that provide arguments to the parameters of a screen $sid$ extend the state $\sigma$, if each construct provides an argument to a parameter, \ie each construct in the sequence extends the state $\sigma$.

As per rule \ref{eq:pVal}, a syntactic construct that provides an argument to a screen's parameter extends the state $\sigma$ in the current configuration with the value $v$, if the value $v$ is provided as an argument to parameter ID $vid$. 

As per rule \ref{eq:pVar}, a syntactic construct that provides an argument to a screen's parameter extends the state $\sigma$ in the current configuration with the value of $x$, if $x$ has the value $v$ in $\sigma$ and $x$ is either a variable, widget or an operation. 

As per rule \ref{eq:pInvVar}, A syntactic construct that provides an argument to a screen's parameter does not change the current configuration if the variable or widget that provides the argument is undefined in $\sigma$.

\begin{table}
    \centering
    \begin{tabular}{c}
    \toprule

    \vbox{\begin{equation}
    \label{eq:andTrue}
    \inferrule
        {\srule{b_0}{\st{sid}{\sigma}}{true} \\
        \srule{b_1}{\st{sid}{\sigma}}{true}}
        {\srule{\bop{b_0}{and}{b_1}}{\st{sid}{\sigma}}{true}}
\end{equation}} \\

    \vbox{\begin{equation}
    \label{eq:andFalse1}
    \inferrule
        {\srule{b_0}{\st{sid}{\sigma}}{false}}
        {\srule{\bop{b_0}{and}{b_1}}{\st{sid}{\sigma}}{false}}
\end{equation}} \\

    \vbox{\begin{equation}
    \label{eq:andFalse2}
    \inferrule
        {\srule{b_1}{\st{sid}{\sigma}}{false}}
        {\srule{\bop{b_0}{and}{b_1}}{\st{sid}{\sigma}}{false}}
\end{equation}} \\

    \vbox{\begin{equation}
    \label{eq:OrTrue1}
    \inferrule
        {\srule{b_0}{\st{sid}{\sigma}}{true}}
        {\srule{\bop{b_0}{or}{b_1}}{\st{sid}{\sigma}}{true}}
\end{equation}} \\

    \vbox{\begin{equation}
    \label{eq:OrTrue2}
    \inferrule
        {\srule{b_1}{\st{sid}{\sigma}}{true}}
        {\srule{\bop{b_0}{or}{b_1}}{\st{sid}{\sigma}}{true}}
\end{equation}} \\

    \vbox{\begin{equation}
    \label{eq:OrFalse}
    \inferrule
        {\srule{b_0}{\st{sid}{\sigma}}{false} \\
        \srule{b_1}{\st{sid}{\sigma}}{false}}
        {\srule{\bop{b_0}{or}{b_1}}{\st{sid}{\sigma}}{false}}
\end{equation}} \\

    \vbox{\begin{equation}
    \label{eq:NotTrue}
    \inferrule
        {\srule{b}{\st{sid}{\sigma}}{false}}
        {\srule{\bop{}{not}{b}}{\st{sid}{\sigma}}{true}}
\end{equation}} \\

    \vbox{\begin{equation}
    \label{eq:NotFalse}
    \inferrule
        {\srule{b}{\st{sid}{\sigma}}{true}}
        {\srule{\bop{}{not}{b}}{\st{sid}{\sigma}}{false}}
\end{equation}} \\
    
    \bottomrule
    \end{tabular}
    \caption{Boolean-related Semantic Rules.}
    \label{tab:bool-related-rules}
\end{table}

\paragraph{Boolean-related rules} The rules listed in \Fref{tab:bool-related-rules} are related to boolean expressions and are similar to traditional notions of conjunction, disjunction, and negation in logical statements. 

\subsection{Example}
We will apply the appropriate semantic rules described above to interpret a portion of the app's storyboard shown in \Fref{fig:story_ex}. Let us assume that the transition from the screen \texttt{Messenger} to the screen \texttt{Contacts} is taken. Consequently, via rule \ref{eq:trVanilla} the existing configuration changes to $\st{sid}{\sigma}$, where $sid$ is the identifier of the \texttt{Contacts} screen.\footnote{Notice that the transition from \texttt{Messenger} to \texttt{Contacts} does not have a boolean expression associated with it. For brevity, we assume that the boolean expression $b$ in rule \ref{eq:trVanilla} evaluates to true.} Since $sid$ in the configuration $\st{sid}{\sigma}$ is the identifier of \texttt{Contacts}, rule \ref{eq:screen} is triggered. Consequently, $\sigma$ is extended with the widgets in the \texttt{Contacts} screen via the widget-related rules. Hence, the configuration is now $\st{sid}{\sigma^\prime}$. At this point, either one of the two outgoing transitions from screen \texttt{Contacts} can be taken via rule \ref{eq:trSeq1} or \ref{eq:trSeq2}. Let us assume that the transition from \texttt{Contacts} to \texttt{SaveStatus} is taken. In this scenario, rule \ref{eq:tr} will be triggered, which changes the configuration to $\st{sid^{\prime\prime}}{\sigma^{\prime\prime}}$. In this new configuration $sid^{\prime\prime}$ is the identifier of the screen \texttt{SaveStatus} and $\sigma^{\prime\prime}$ contains the value of screen \texttt{SaveStatus}'s input parameter $x$. Thereafter, rule \ref{eq:screenNoTran} is triggered and $\sigma^{\prime\prime}$ is extended with the widgets in \texttt{SaveStatus} via the widget-related rules. Since there are no outgoing transitions from \texttt{SaveStatus}, the app stays in \texttt{SaveStatus} via rule \ref{eq:appSid} or is terminated via rule \ref{eq:appStop}.     

\subsection{Safety and Progress}

We say that an app is safe if the app does not reach an \textit{invalid configuration} or a \textit{terminal configuration}. Since there are no invalid configurations, an app does not reach an \textit{invalid configuration}. As far as terminal configuration is concerned, we define it as the $\st{\perp}{\sigma}$, such that $\forall x \in dom(\sigma) : \sigma(x) = \perp$. The only time an app reaches the terminal configuration is when the app is stopped by the user or Android. In this configuration no other transition is possible and the app is said to be terminated. \textit{Since no rule defined above, except rule \ref{eq:appStop}, makes an app reach the terminal configuration, we say that the semantics ensure safety by construction}.

The semantic rules defined above ensure that an app is never "stuck", \ie if an app is in a valid configuration or a non-terminal configuration, then the app will either terminate or will be evaluated further in some configuration. We call this property of not getting stuck as progress. Formally, progress is specified as the following theorem:       

\begin{theorem}
\label{theo:appSafety}
If an app is in a configuration $\st{sid}{\sigma}$, then $\exists \st{sid^\prime}{\sigma^\prime} : \ruleExp{app}{\st{sid}{\sigma}}{\langle app,\st{sid^\prime}{\sigma^\prime} \rangle}$, or $\ruleExp{app}{\st{sid}{\sigma}}{\st{\perp}{erase(\sigma,\sigma)}}$ 
%If an app is in $\st{sid}{\sigma}$, then the app will always move to a configuration $\st{sid\prime}{\sigma\prime}$ or the terminal configuration $\st{\perp}{\sigma\prime\prime}$, such that $\forall x \in dom(\sigma\prime\prime): \sigma\prime\prime(x) = \perp$.
\end{theorem}

We need the following lemma to prove theorem \ref{theo:appSafety}.

\begin{lemma}
\label{lem:trans}
If an app is in a configuration $\st{sid}{\sigma}$ and $\forall tr \in TR : outTr(sid) = tr$, then $\exists (sid^\prime, \sigma\prime) : \srule{tr}{\st{sid}{\sigma}}{\st{sid^\prime}{\sigma^\prime}}$, where $(sid^\prime, \sigma\prime)$ may or may not equal $\st{sid}{\sigma}$.   
\end{lemma}

\begin{proof}
We prove Lemma \ref{lem:trans} by induction on the structure of $tr$.

\textit{Base Case}: When $tr$ is one transition from the screen $sid$, Lemma \ref{lem:trans} is trivially true since either Rule \ref{eq:tr},  \ref{eq:trVanilla}, \ref{eq:trSelf}, \ref{eq:trFail1}, or \ref{eq:trFail2} will always apply.

\textit{Inductive Case}: When $tr$ is an ordered set of transitions, $tr_0tr_1$, from the screen $sid$, then either apply rule \ref{eq:trSeq1} and \ref{eq:trSeq2} with the induction hypothesis on $tr_0$ or apply \ref{eq:trSeq2} with the induction hypothesis on $tr_1$. 

Hence, by the principle of induction, Lemma \ref{lem:trans} is true.

\end{proof}

\begin{proof}

Let us assume that an app is in $\st{sid}{\sigma}$, where $sid$ is a screen in the app and $\sigma$ is the state of the app at screen $sid$.

Assume that there is an ordered set of transitions $tr$ from screen $sid$, then by lemma \ref{lem:trans}, $\exists (sid^\prime, \sigma^\prime)$ such that the app will be further evaluated in $(sid^\prime, \sigma^\prime)$ as per rule \ref{eq:appSid}. If the app is stopped, then the app will terminate as per rule \ref{eq:appStop}. 
\end{proof}

\section{Information Flow Analysis Algorithm}
\label{sec:infoFlowFormal}
This analysis tracks the flow of information in the form of values from \textit{sources} to \textit{sinks} in a storyboard, as defined in \Fref{sec:analysis}. In this section, We formally describe the algorithm. 

\paragraph{Meta-language:} Before specifying the algorithm, we define the following functions and structures required to understand it and the corresponding proof of correctness.

\begin{enumerate}
    \item $Idn : Exp \rightarrow IDN$, where $Exp = W \cup P \cup F$ and $IDN = WID \cup VID \cup FN$, is a function is used to obtain the identifier of a widget, input parameter, or an operation. It is defined as follows:
        \begin{enumerate}
            \item \textit{Widget} : if $w \in Exp$ and $w ::=\widget{wt}{wid}{x}$, then $Idn(w) = gen(wid,\phi(wid))$
            \item \textit{Parameter} : if $p \in Exp$ and $p ::= \param{vid}{x}$, then $Idn(p) = gen(vid,\phi(vid))$
            \item \textit{Operation} : if $\textit{f} \in Exp$ and $\textit{f} ::= \fn{fn}{rn.cn}{a}$, then $Idn(\textit{f}) = fn$
        \end{enumerate}

    \item $FV: Exp \rightarrow 2^{IDN}$, is a function that maps $e \in Exp$ to the power set of $IDN$. $FV$ is used to retrieve the identifiers that are used as values of a widget, input parameter, or operation. It is defined as follows:
    
    \begin{enumerate}
        \item if $e \in F$, then
        \begin{enumerate}
            \item if $e ::= \fn{fn}{rn.cn}{vid}$, then \textit{FV}$(e) = \{gen(vid,\phi(vid))\}$
            \item if $e ::= \fn{fn}{rn.cn}{wid}$, then \textit{FV}$(e) = \{gen(wid,\phi(wid))\}$
            \item if $e ::= \fn{fn}{rn.cn}{f\prime}$, then \textit{FV}$(e) = \{Idn(f\prime)\}$
            \item if $e ::= \fn{fn}{rn.cn}{a_0a_1}$, then \textit{FV}$(e) = $ \textit{FV}$(f_0)$ $\cup$ \textit{FV}$(f_0)$, \\ where $f_0 ::= \fn{fn}{rn.cn}{a_0}$ and $f_1 ::= \fn{fn}{rn.cn}{a_1}$ 
            \item otherwise, \textit{FV}$(e) = \emptyset$
        \end{enumerate}
    
        \item if $e \in W$, then
            \begin{enumerate}
                \item if $e ::= \widget{wt}{wid}{vid}$, then \textit{FV}$(e) = \{gen(vid,\phi(vid))\}$
                \item if $e ::= \widget{wt}{wid}{f}$, then \textit{FV}$(e) = \{Idn(f)\}$
                \item otherwise, \textit{FV}$(e) = \emptyset$
            \end{enumerate}
            
            \item if $e \in P$, then
            \begin{enumerate}
                \item if $e ::= \param{vid_0}{vid_1}$, then \textit{FV}$(e) = \{gen(vid_1,\phi(vid_1))\}$
                \item if $e ::= \param{vid}{wid}$, then \textit{FV}$(e) = \{gen(wid,\phi(wid))\}$
                \item if $e ::= \param{vid}{f}$, then \textit{FV}$(e) = \{Idn(f)\}$
                \item otherwise, \textit{FV}$(e) = \emptyset$
            \end{enumerate}
    \end{enumerate}
    
    % \item The operator $@ : (\textbf{IDN},\textbf{IDN} \rightarrow \textbf{IDN})$ takes a pair $(idn_1,idn_2)$ such that $idn_1, idn_2 \in \textbf{IDN}$ and returns a new $idn \in \textbf{IDN}$.
    
    \item The \textit{safe} relation is used to collect direct flows between sources and sinks in the storyboard that have been marked by a developer with the \textit{safe} keyword. \Fref{sec:syntax} shows the syntax for marking direct flows as \textit{safe}.
    
    \item \textit{safeParams} denotes the set of input parameters of a proxy screen $ps \in PS$ such that $ps$ has been marked \textit{safe} or $ps$ has the \textit{app} attribute set.
    
    \item $CV : F \rightarrow BV$ is a predicate such that $CV(f) = true$ if operation $f$ is an untrusted sink; $CV(f) = false$ otherwise.
    
    \item $IV : F \rightarrow BV$ is a predicate such that $IV(f) = true$ if operation $f$ is an untrusted source; $IV(f) = false$ otherwise.
    
\end{enumerate}

\paragraph{\textbf{Algorithm}:} The analysis algorithm proceeds in two stages. First, it captures the \textit{direct} relationships between the identifier of a widget, an input parameter, or an operation with the identifier used as value in a widget, input parameter, or an argument to an operation in a binary relation called $influences$. It then uses \textit{influences} to build a reflexive transitive closure $influences^*$ to capture the \textit{indirect} relationships between widgets, input parameters, and operations. Second, for each widget, operation, and input parameter expression, the analysis collects their identifiers if they are directly or indirectly influenced by unsafe and untrusted identifiers. A violation is detected if a widget, input parameter, or operation has a non-empty set of such identifiers associated with it. 

The algorithm is formally defined in the following steps:

\begin{enumerate}
    \item
    \label{def:influences}
    Calculate the binary relation \textit{influences}: $IDN \rightarrow IDN$ for each $e \in W \cup P \cup F$ as follows:
    \begin{enumerate}
        \item \textit{Widgets} $w \in W$:
        \begin{enumerate}
            \item if $w ::= \widget{wt}{wid}{vid}$, then $\{(gen(vid,\phi(vid)),gen(wid,\phi(wid)))\}$.
            \item if $w ::= \widget{wt}{wid}{f}$, then $\{(Idn(f),gen(wid,\phi(wid)))\}$.
            \item if $w ::= w_0w_1$, then $\textit{influences}_{w0} \cup \textit{influences}_{w1} \cup \textit{influences}$
        \end{enumerate}
        
        \item \textit{screen input parameters} $p \in P$:
            \begin{enumerate}
                \item if $p ::= \param{vid_0}{vid_1}$, then $\{(gen(vid_1,\phi(vid_1)),gen(vid_0,\phi(vid_0)))\}$
                \item if $p ::= \param{vid}{wid}$, then $\{(gen(wid,\phi(wid)),gen(vid,\phi(vid)))\}$
                \item if $p ::= \param{vid}{f}$, then $\{(Idn(f),gen(vid,\phi(vid)))\}$
                \item if $p ::= p_0p_1$, then $\textit{influences}_{p0} \cup \textit{influences}_{p1} \cup \textit{influences}$
            \end{enumerate}
            
        \item \textit{Operations} $f \in F$:
            \begin{enumerate}
                \item if $f ::= \fn{fn}{rn.cn}{vid}$, then $\{(gen(vid,\phi(vid)),fn)\}$
                \item if $f ::= \fn{fn}{rn.cn}{wid}$, then $\{(gen(wid,\phi(wid)),fn)\}$.
                \item if $f ::= \fn{fn}{rn.cn}{f\prime}$, then $\{(Idn(f\prime),fn)\}$
                
                \item if $f ::= \fn{fn}{rn.cn}{a_0a_1}$, then $\textit{influences}_{f0} \cup \textit{influences}_{f1} \cup \textit{influences}$, \\ where $f_{0} ::= \fn{fn}{rn.cn}{a_0}$ and $f_{1} ::= \fn{fn}{rn.cn}{a_1}$
            \end{enumerate}
    \end{enumerate}
    
    \item Calculate \textit{influences}$^*$, the reflexive transitive closure of the binary relation \textit{influences}.
    
    % \item Calculate the predicate $CIV : F \rightarrow T$ for every operation in the storyboard such that $CIV(f) = t$ if operation $f$ uses a resource to produce/consume data that cannot be trusted/sensitive (\eg EXT\_STORE).  
    
    \item Mark untrusted sources using the \textit{Untrusted} : $IDN \rightarrow BV$ predicate for each $e \in F \cup U$ as follows:
    
    \begin{enumerate}
        \item if $e \in F \land IV(e)$, then $\forall (Idn(f),k) \in \textit{influences}^* : \textit{Untrusted}(k) = \textit{true}$
        \item if $e ::= us/K$, then $\forall k \in K : \textit{Untrusted}(k) = \textit{true} \land \forall (gen(vid,\phi(vid),y) \in \textit{influences}^* : \textit{Untrusted}(y) = \textit{true}$
    \end{enumerate}
    
    \item Collect flows marked \textit{safe} for each $e \in W \cup F \cup P$ as follows:
    
    \begin{enumerate}
        \item if $e ::= \textit{safe}$ $\widget{wt}{wid}{x}$, then $ \{x,gen(wid,\phi(wid)))\} \cup \textit{safe}$ and $ \forall (gen(wid,\phi(wid)),k) \in \textit{influences}^* : \{(gen(wid,\phi(wid)),k)\} \cup \textit{safe}$
        
        \item if $e ::= w_0w_1$, then $ \textit{safe}_{w0} \cup \textit{safe}_{w1} \cup safe$
        
        \item if $e ::= \fn{fn}{rn.cn}{\textit{safe}\hspace{1mm}x}$, then $ \{(gen(x,\phi(x)),fn)\} \cup \textit{safe}$, where $x \in (WID \cup VID)$
        
        \item if $e ::= \fn{fn}{rn.cn}{\textit{safe}\hspace{1mm}f}$, then $ \{(Idn(f),fn)\} \cup \textit{safe}$
        
        \item if $e ::= \fn{fn}{rn.cn}{a_0a_1}$, then $ \textit{safe}_{e0} \cup \textit{safe}_{e1} \cup \textit{safe}$, where $e0 ::= \fn{fn}{rn.cn}{a_0}$ and $e1 ::= \fn{fn}{rn.cn}{a_1}$
        
        \item if $e ::= \param{vid}{x}$ and $gen(vid,\phi(vid)) \in safeParams$, then $\{(x,gen(vid,\phi(vid))) \in \textit{safe}\}$
        
        \item if $e ::= p_0p_1$, then $ \textit{safe}_{p0} \cup \textit{safe}_{p1} \cup \textit{safe}$
    \end{enumerate}
    
    \item Calculate the function $\textit{IF} : Exp \rightarrow 2^{IDN}$ for each $e \in W \cup F \cup P$ as follows:
    
    \begin{enumerate}
    \item if $e \in W$, then
    \begin{enumerate}
        \item if $e ::= \widget{wt}{wid}{vid}$, then \\ $\textit{IF}(e) = \{gen(wid,\phi(wid)) \dash \textit{Untrusted}(vid) 
        \land (((vid,wid) \notin \textit{safe}) \land (\exists (z,vid): (z,vid) \in \textit{influences} \implies (\forall (y,z): (y,z) \in \textit{influences}^* \implies (y,z) \notin \textit{safe})))$
        
        \item if $e ::= \widget{wt}{wid}{f}$, then \\ $\textit{IF}(e) = \{gen(wid,\phi(wid)) \dash IV(f) \land (Idn(f),wid) \notin \textit{safe}$
        
        \item if $e ::= w_0w_1$, then $\textit{IF}(e) = \textit{IF}(w_0) \cup \textit{IF}(w_1)$
    \end{enumerate}
    
    \item if $e \in F$, then
    \begin{enumerate}
        \item if $e ::= \fn{fn}{rn.cn}{wid}$, then \\
        $\textit{IF}(e) = \{\textit{fn} | (wid,\textit{fn}) \notin \textit{safe} \land Idn(w) = wid \land \textit{IF}(w) \neq \emptyset \} \\ \cup 
        \{\textit{fn} | CV(e) \land ((wid,\textit{fn}) \notin safe) \land (\exists (\textit{fn}^\prime,wid) : (\textit{fn}^\prime,wid) \in \textit{influences} \implies (\textit{fn}^\prime,wid) \notin \textit{safe}) \land (\exists (y,vid), (vid,wid) : ((y,vid) \in \textit{influences} \land (vid,wid) \in \textit{influences}) \implies (\forall (k,vid) : (k,vid) \in \textit{influences}^* \implies (k,vid) \notin \textit{safe}))\}$
        
        \item if $e ::= \fn{fn}{rn.cn}{vid}$, then \\
        $\textit{IF}(e) = \{\textit{fn} | (\textit{Untrusted}(vid) \lor CV(e)) \land (vid,\textit{fn}) \notin \textit{safe} \land (\exists (z,vid) : (z,vid) \in \textit{influences} \implies (\forall (y,z) : (y,z) \in \textit{influences}^* \implies (y,z) \notin \textit{safe}))\}$
        
        \item if $e ::= \fn{fn}{rn.cn}{f}$, then \\
        $\textit{IF}(e) = \{\textit{fn} | (IV(f) \lor CV(e)) \land (Idn(f),\textit{fn}) \notin \textit{safe}\}$
        
        \item if $e ::= \fn{fn}{rn.cn}{a_0a_1}$, then \\
        $\textit{IF}(e) = \textit{IF}(e_0) \cup \textit{IF}(e_1)$, where $e_0 ::= \fn{fn}{rn.cn}{a_0}$ and $e_1 ::= \fn{fn}{rn.cn}{a_1}$
    \end{enumerate}
    % $IF(e) = \{Idn(e) \dash \exists idn \in FV(e) : (Untrusted(idn) \lor \exists idn\prime \in \textbf{IDN} : (idn\prime,idn) \in influences^* \land Untrusted(idn\prime) \land \neg safe(idn\prime)) \land \neg safe(Idn(e)@idn)\} \cup 
    % \{Idn(e) \dash Untrusted(Idn(e)) \land (\exists i \in FV(e): (\neg safe(Idn(e)@i) \land (\exists idn \in \textbf{IDN} : (idn,i) \in influences^* \land \neg safe(idn)))\}$
    
    \item if $e \in P$, then
    \begin{enumerate}
        \item if $e :: = \param{vid}{id}$, where $id \in (WID \cup VID)$, then \\
        $\textit{IF}(e) = \{gen(vid,\phi(vid)) \dash (id,vid) \in \textit{safe} \lor (\forall (y,id) \in \textit{influences}^*: (y,id) \notin \textit{safe})\}$
        
        \item if $e :: = \param{vid}{f}$, then
        $\textit{IF}(e) = \{gen(vid,\phi(vid)) \dash Idn(f) = \textit{fn} \land (\textit{fn},vid) \in \textit{safe}$
        
        \item if $e ::= e_0e_1$, then
        $\textit{IF}(e) = \textit{IF}(e_0) \cup \textit{IF}(e_1)$
    \end{enumerate}
    
\end{enumerate}
\end{enumerate}

The analysis reports a violation if $\exists e \in Exp: \textit{IF}(e) \neq \emptyset$ in the storyboard.

As an example, consider information flow analysis of the storyboard in \Fref{fig:story_ex}. As per the algorithm, we first build the binary relation: 

\begin{multline*}
    \textit{influences} = \{(y,Phone), (Phone,z), (Phone,x), (dispMsg, Status), (x,dispMsg), \\ (Phone,savePhone), (getContacts,sendMsg)\}
\end{multline*}

Next, a reflexive transitive closure is calculated as follows:

\begin{multline*}
    \textit{influences}^* = \{(y,Phone), (Phone,z), (Phone,x), (dispMsg, Status), (x,dispMsg), \\ (Phone,savePhone), getContacts,sendMsg), (y,z), (y,x), (Phone,dispMsg), \\ (y,dispMsg), (y,Status), (x,Status)), ...\}
\end{multline*}

The predicates \textit{IV} and \textit{CV} return \textit{false} for every operation in the storyboard since none of them use resources that correspond to untrusted source/sink. Since, $y$ is a variable from an external app, $\forall (y,k) \in \textit{influences}^* : \textit{Untrusted}(y,k) = \textit{true}$.

\Fref{tab:info-flow-ex} shows the result of computing \textit{IF} for each $e\in Exp$ in the storyboard before any flow was marked as \textit{safe}. Since $\exists e \in Exp : \textit{IF}(e) \neq \emptyset$, the analysis reports a violation. To fix the violation, the variable $y$ needs to be obtained from a trusted source or the flow/s related to $y$ should be marked \textit{safe}. If we take the latter approach and mark the flow between $y$ and the widget $Phone$ as \textit{safe} as shown in \Fref{fig:story_ex}, then on running the analysis again, \textit{safe} is as follows:

\begin{equation*}
    \textit{safe} = \{(y,Phone), (Phone,z), (Phone,x), ...\}
\end{equation*}

Since $\forall e \in Exp : \textit{IF}(e) = \emptyset$ as shown in \Fref{tab:info-flow-ex} third column, no violations are reported by the analysis.

\begin{table}
  \centering
  \ifdef{\TopCaption}{
    \caption{Information Flow Analysis}
  }{}
  \begin{tabular}{@{}ccc@{}} 
    \toprule
    \textit{Exp} & \textit{IF} & \textit{IF} (safe)\\
    \midrule
    TextView Phone y & \{gen(Phone,$\phi(Phone)$)\} & \{\}\\
    TextView Status "MsgSent" & \{\} & \{\}\\
    TextView Status dispMsg(x) & \{\} & \{\} \\
    sendMsg(...) & \{\} & \{\}\\
    getContacts(...) & \{\} & \{\}\\
    dispMsg(x) & \{Idn(dispMsg)\} & \{\}\\
    param z Phone & \{\} & \{\}\\
    param x Phone & \{\} & \{\} \\
    \bottomrule
  \end{tabular}
  \ifundef{\TopCaption}{
    \caption{Information Flow Analysis of Storyboard shown in \Fref{fig:story_ex}. The first column indicates an $e \in Exp$ (\ie widget, operation, or input parameter of a screen). The second column indicates a set obtained by evaluating IF(e) assuming that nothing was marked as \textit{safe}. The third column indicates IF with a flow marked as \textit{safe}.}
  }{}
  \label{tab:info-flow-ex}
\end{table}

\subsection{Proof of Correctness}

The correctness of the function \textit{IF} hinges on the calculation of the transitive closure $\textit{influences}^*$, predicate \textit{Untrusted}, and set \textit{safe}. Since \textit{Untrusted} and \textit{safe} are based on a pre-defined list of trusted and untrusted identifiers and developer-provided annotations/indicators, it is enough to prove that the relations captured in $\textit{influences}^*$ reflect the flow as specified in the semantics.

\begin{theorem}
\label{theo-DI2}
 $\forall a \in APP$, $\forall sid, sid^\prime \in SID, x,y \in ID, \sigma,\sigma^\prime$ : $\ruleExpL{a}{\st{sid}{\sigma}}{\langle a, \st{sid^\prime}{\sigma^\prime} \rangle}$ $\land$ $\sigma(x) = \sigma^\prime(y)$ $\implies$ $(x,y)$ $\in$ \textit{influences*}
\end{theorem}

\begin{proof}
We prove this by induction on the no. of steps it takes an app from the current configuration $(sid,\sigma)$ to reach a configuration $(sid^\prime,\sigma^\prime$).

\textit{Base Case}: Assume that $a$ is an app, $sid,sid^\prime \in SID, x,y \in ID, \sigma,\sigma^\prime$ such that $\ruleExp{a}{\st{sid}{\sigma}}{\langle a, \st{sid^\prime}{\sigma^\prime}} \rangle$ $\land$ $\sigma(x) = \sigma^\prime(y)$. In screen $sid^\prime$, $\exists l$ such that $l$ is a widget or provides an argument to an input parameter of another screen as part of a transition from $sid^\prime$. Since $\sigma(x) = \sigma^\prime(y)$, $y = Idn(l)$ and $k \in FV(l)$, where either $k = x$ or $k$ is the identifier of an operation with an argument $x$. By the definition of \textit{influences}, either $(x,y) \in \textit{influences}$ or $(x,k) and (k,y) \in \textit{influences}$. In either case, $(x,y) \in \textit{influences*}$.

Therefore, \ref{theo-DI2} is true.\\

\textit{Inductive Case}: 
\begin{itemize}
    \item{Case 1}: Let's assume that $a$ is an app, $sid,sid^\prime \in SID, x,z \in ID, \sigma,\sigma^{\prime\prime}$ such that  $\ruleExpK{a}{\st{sid}{\sigma}}{\langle a, \st{sid^\prime}{\sigma^{\prime\prime}} \rangle}{k}$ $\land$ $\sigma(x) = \sigma^{\prime\prime}(z)$ $\implies$ $(x,z)$ $\in$ \textit{influences*}, where $k$ is the no. of steps it takes to reach from $\sigma$ to $\sigma^{\prime\prime}$. This is our induction hypothesis.
    
    Further, assume that $\exists sid^{\prime\prime} \in SID, y \in ID, \sigma^\prime$ such that $\ruleExp{a}{\st{sid^\prime}{\sigma^{\prime\prime}}}{\langle a, \st{sid^{\prime\prime}}{\sigma^\prime}} \land \sigma^{\prime\prime}(z) = \sigma^\prime(y)$ and $l$ is a widget in screen $sid^{\prime\prime}$ or provides an argument to the input parameter of a screen via a transition from $sid^{\prime\prime}$. Since $\sigma^{\prime\prime}(z) = \sigma^\prime(y)$, $y = Idn(l)$ and $k \in FV(l)$, where $k = z$ or $k$ is the identifier of an operation that uses $z$ as an argument. In either case, $(z,y) \in \textit{influences}$ due to the definition of \textit{influences}. Since $\textit{influences} \subseteq \textit{influence*}$, $(z,y) \in \textit{influences*}$. From the induction hypothesis, $(x,z) \in \textit{influences*}$. Since, \textit{influences*} is a transitive closure, $(x,y) \in \textit{influences*}$. 
    
    \item{Case 2}: Let's assume that $a$ is an app, $sid,sid^\prime \in SID, \sigma,\sigma^\prime$ such that $\ruleExpK{a}{\st{sid}{\sigma}}{\langle a, \st{sid^\prime \rangle}{\sigma^\prime}}{m}$ $\land$ $\sigma(x) = \sigma^\prime(z)$ $\implies$ $(x,z)$ $\in$ \textit{influences*}, and $sid^{\prime\prime} \in SID, \sigma^{\prime\prime}$ such that $\ruleExpK{a}{\st{sid^\prime}{\sigma^\prime}}{\langle a, \st{sid^{\prime\prime}}{\sigma^{\prime\prime}} \rangle}{n}$ $\land$ $\sigma^\prime(z) = \sigma^{\prime\prime}(j)$ $\implies$ $(z,j)$ $\in$ \textit{influences*}, and $k = m+n$. 
    
    Further, assume that $\exists sid^{\prime\prime\prime} \in SID, y \in ID, \sigma^{\prime\prime\prime}$ : $\ruleExp{a}{\st{sid^{\prime\prime}}{\sigma^{\prime\prime}}}{\langle a,\st{sid^{\prime\prime\prime}}{\sigma^{\prime\prime\prime}}} \land \sigma^{\prime\prime}(j) = \sigma^{\prime\prime\prime}(y)$ and $l$ is a widget in screen $sid^{\prime\prime\prime}$ or provides the argument to an input parameter of a screen via a transition from $sid^{\prime\prime\prime}$. Since $\sigma^{\prime\prime}(j) = \sigma^{\prime\prime\prime}(y)$, $y = Idn(l)$ and $k \in FV(l)$, where $k = j$ or $k$ is the identifier of an operation that uses $j$ as an argument. In either case, $(j,y) \in \textit{influences}$ by the definition of \textit{influences}. Since $\textit{influences} \subseteq \textit{influence*}$, $(j,y) \in \textit{influences*}$. From the induction hypothesis, $(x,z) \in \textit{influences*}$ and $(z,j) \in \textit{influences*}$. Since, \textit{influences*} is a transitive closure, $(x,y) \in \textit{influences*}$. 
\end{itemize}
 
\textit{Conclusion}: By the principle of induction \ref{theo-DI2} is true.

\end{proof}

The converse of Theorem \ref{theo-DI2} does not hold for the analysis because the analysis does not consider the effects of constraints when building the \textit{influences} relation. So, it is possible that an ID $x$ \textit{influences} $y$ even if the semantics does not allow $x$ to flow into $y$. For example, let us assume that a transition $t$ is guarded by a constraint $b$ and when $b$ is \textit{true}, an argument $a$ is passed to the destination screen. Let's also assume that $x \in \textit{FV}(a)$ and $y = Idn(a)$. Finally, let us assume that $b$ is always \textit{false}. In such a scenario, $x$ \textit{influences} $y$ but the semantics will not allow $x$ to flow into $y$. This implies that the analysis will flag violations even if there isn't any. However, the developer can override the violation by setting the \textit{safe} attribute appropriately.

\end{document}